\documentclass{lmcs}
\pdfoutput=1

\usepackage{lastpage}
\lmcsdoi{21}{3}{15}
\lmcsheading{}{\pageref{LastPage}}{}{}%
{May~17,~2024}{Aug.~07,~2025}{}

\usepackage{amssymb, amsmath, amsfonts, amsthm, latexsym, graphics, graphicx}
\usepackage{tikz}
\usepackage{subcaption}
\usepackage[utf8]{inputenc}
\usepackage{booktabs}

\usetikzlibrary{fit,calc,positioning,decorations.pathreplacing,matrix}
\usetikzlibrary{shapes}
\usetikzlibrary{arrows,positioning}
\usetikzlibrary{automata}
\usepackage{xcolor}

\numberwithin{theorem}{section}
\definecolor{macouleur}{RGB}{4,139,154}

\title[Learning Weighted Automata]{Feasibility of Learning Weighted Automata on a Semiring}

\author{Laure Daviaud\lmcsorcid{0000-0002-9220-7118}}[a]
\address{University of East Anglia, UK}
\author{Marianne Johnson\lmcsorcid{0000-0003-4059-845X}}[b]
\address{University of Manchester, UK}


\begin{document}
\begin{abstract}
\noindent Since the seminal work by Angluin and the introduction of the $L^*$-algorithm, active learning of automata by membership and equivalence queries has been extensively studied to learn various extensions of automata. For weighted automata, algorithms for restricted cases have been developed in the literature, but so far there was no global approach or understanding how these algorithms could apply (or not) in the general case. In this paper we chart the boundaries of the Angluin approach. We use a class of \textit{hypothesis automata} which are constructed, in Angluin's style, by using membership and equivalence queries and solving certain finite systems of linear equations over the semiring, and we show the theoretical limitations of this approach. We classify functions with respect to how `guessable' they are, corresponding to the existence of hypothesis automata computing a given function, and how such an hypothesis automaton can be found. Of course, from an algorithmic standpoint, knowing that a solution (hypothesis automaton) exists need not translate into an effective algorithm to find one. We  relate our work to the existing literature with a discussion of some known properties ensuring algorithmic solutions, illustrating the ideas over several familiar semirings (including the natural numbers).
\end{abstract}
\maketitle
	
\section{Introduction}\label{section:intro}
Imagine the following situation: you want to model a computer system (for verification purposes for example) with some mathematical abstraction, but the internal state of the system cannot be accessed. Think for example of very complex systems such as AI systems, that can realistically only be viewed as black-boxes: the system can be tested, you can feed it some chosen inputs and observe the outputs. Given these (finitely many)  input-output pairs, you create a model. Suppose you also have a way to decide whether your model matches (or, more realistically in the application fields, is close enough to) the system's behaviour. If this is not the case, the system gives you an input that behaves incorrectly in your model and using this new input-output pair together with further observations, you construct a new model to test. This is active learning.   

\subsubsection*{\textbf{Automata learning:}} In her seminal paper~\cite{A87}, Angluin introduced the $L^*$-algorithm where a learner tries to guess a rational language $L$ known only by an oracle. The learner is allowed to ask two types of queries: (1) membership queries where the learner chooses a word $w$ and asks the oracle whether $w$ belongs to $L$, and (2) equivalence queries where the learner chooses an automaton $\mathcal{A}$ and asks the oracle whether $\mathcal{A}$ recognises $L$; if this is not the case, the oracle gives a word (called counter-example) witnessing this fact. The $L^*$-algorithm allows the learner to correctly guess  the minimal deterministic automaton recognising $L$ in a  number of membership and equivalence queries polynomial in its size and the length of counter-examples given by the oracle. An overview of applications of automata learning can be found in~\cite{L06,HS16,V17}. In particular, the $L^*$-algorithm has been used in~\cite{WGY18} to extract automata from recurrent neural networks.

\subsubsection*{\textbf{Weighted automata:}} Weighted automata are a quantitative extension of automata where transitions are weighted in a semiring, allowing to model probabilities, costs or running time of programs~\cite{Schutzenberger1961,Droste2009,DrosteK21}. They find applications in image compression~\cite{CK93,Culik1997}, language and speech processing~\cite{Mohri1997,MPR05}, bioinformatics~\cite{AMT08}, formal verification~\cite{Chatterjee2010,AKL11}, analysis of on-line algorithms \cite{Aminof2010} and probabilistic systems~\cite{Vardi85}. Notably, they have also been used recently to model recurrent neural networks~\cite{RLP19,OWSH20,ZDXML021,WZS22}.

\subsubsection*{\textbf{Extending automata learning:}} Active learning has been extended to weighted automata but only for restricted classes: it has first been studied for fields~\cite{BV96,BM15}, and then principal ideal domains~\cite{HKRS20}. The Angluin framework has also been generalised to other extensions of automata and transducers~\cite{V96,V00,BHKL09,BLN12,AF14,AEF15,BLN16,DD17,MSSKS17,HS020}. In~\cite{CPS21} (see also \cite{HSS17} \cite{BKR19} and \cite{US2020} for earlier works), a categorical approach was proposed to encompass many of these cases. However, although this encompasses the case of weighted automata over fields, this framework does not include weighted automata in general. As often with weighted automata, finding a unified approach (working for all semirings) is not easy.

In~\cite{HKRS20}, van Heerdt, Kupke, Rot and Silva set a general framework for learning weighted automata, present a general learning algorithm, and focus on specifying two conditions on the semirings (ascending chain condition and progress measure), ensuring termination of the algorithm - this allows them to develop a learning algorithm \`a la Angluin for automata weighted on a principal ideal domain. More recently, a polynomial-time algorithm within this general framework has been developed in \cite{BCSW24} for the principal ideal domain $\mathbb{Z}$. Section 5 and Theorem 17 of~\cite{HKRS20} also present a specific example of an automaton (as well as a more general class) over the non-negative integers for which their algorithm does not terminate. In the present paper, we develop a general framework that gives a deeper understanding of the reasons for that. 

\subsubsection*{\textbf{Contributions:}} Our approach is somewhat different to what has been done so far: we look globally at all functions computed by finite state weighted automata and investigate, for each of them, whether they \emph{can potentially be learned} by an Angluin-style algorithm, i.e. where the automata that are guessed (hypothesis automata, see Definition~\ref{definition:hankel}) are constructed in a prescribed way using finitely many membership queries (to be explained in detail below). Working over an arbitrary semiring, we define levels of guessability and a hierarchy of functions with respect to these levels: from the ones for which no such guess will be correct to the ones where there is a simple theoretical strategy to construct a correct guess. Our main result in this regard (Theorem \ref{prop:weakly}) is a syntactic condition on the functions that have the potential to be learnt. For such a function (weakly guessable functions, see Definition~\ref{definition:weakly}) there must exist a finite state automaton of a prescribed type computing it: we call these automata \emph{literal} (Definition~\ref{definition:literal}) and they are closely related to residual automata.  In cases already widely studied in the literature (e.g. where the weights are drawn from fields or integral domains), it turns out that \emph{every} function computed by a finite state automaton can be computed by a literal automaton. Thus, in the existing literature, focus is (reasonably) on algorithmic issues. In this paper, we focus on picturing the landscape for active learning in arbitrary semirings, and showing the limitations; for instance over the non-negative integers,  it is easily seen that the example of van Heerdt, Kupke, Rot and Silva mentioned in the previous paragraph cannot be computed by a literal automaton. 

Within this framework, we then consider some properties of semirings that cause certain parts of our general hierarchy to collapse, and consider several specific well-studied examples of semirings (including the non-negative integers) in this context and provide a more detailed picture of the hierarchy for these (Section~\ref{section:hierarchy}). Note that the learning process is based on equivalence queries, and one might ask whether the undecidability of equivalence for certain semirings~\cite{Krob92} is problematic. Since in the application fields, equivalence can rarely be used as such (and often testing or approximation might be used instead), we believe having a framework that applies even in these cases is still of interest.

The difficulty of this work was in setting a new framework and the (we believe) right notions for it. Once the correct definitions are in place, the results and proofs follow nicely. 

\subsubsection*{\textbf{Structure of the paper:}} In Section~\ref{section:prelim}, we recall the definition of weighted automata, some specific examples, and outline the type of learning we are considering. We end this section with a set of questions that arise naturally, and which motivate the definitions and results in the rest of the paper. In Section~\ref{section:guessable}, we explain what we mean by levels of guessability and define classes of functions depending on them. We give several characterisations for these classes of functions, and relate our work back to the existing literature (specifically, ~\cite{HKRS20}), which focuses on algorithmic issues in specific settings. In Section~\ref{section:hierarchy}, we picture the hierarchy based on the previous definitions, prove its strictness and give conditions which cause parts of the hierarchy to collapse. Throughout we raise a number of questions and indicate directions for future progress, summarising these in Section \ref{section:conclusion}.

\section{Learning weighted automata}
\label{section:prelim}

In this section, we recall the definition of weighted automata (on a semiring) and the type of learning we are considering.

\subsection{Weighted automata and specific semirings}
We begin by recording some preliminary definitions and notation. A finite alphabet $\Sigma$ is a finite set of symbols, called letters.  A finite word is a finite sequence of letters and $\varepsilon$ denotes the empty word. We write $\Sigma^*$ to denote the set of all words over the alphabet $\Sigma$, including the empty word. We will also denote by $|w|$ the length of a word $w$ and by $|w|_a$ the number of occurrences of the letter $a$ in $w$.

A monoid $(M,\ast,e)$ is a set $M$ equipped with a binary operation $\ast$ that is associative and has a neutral element $e$. For example, the set $\Sigma^*$ together with the binary operation of concatenation of words and the empty word $\varepsilon$ is a monoid; it is the free monoid on the set $\Sigma$. A semiring $(S,\oplus,\otimes,0_S,1_S)$ is a set $S$ equipped with two operations such that $(S,\oplus,0_S)$ and $(S,\otimes,1_S)$ are monoids, $\oplus$ is commutative, $\otimes$ distributes on the left and the right of $\oplus$ and $0_S$ is a zero for $\otimes$.  A semiring is said to be commutative if $\otimes$ is commutative.

Recall that for a set $Y$, we write $S^Y$ to denote the set of all functions from $Y$ to $S$. For non-empty (but possibly infinite) countable sets $I$ and $J$ we say that $A \in S^{I \times J}$ is a matrix, with rows indexed by the set $I$ and columns indexed by the set $J$. For $A \in S^{I \times J}$ and $B \in S^{J \times K}$ where $J$ is finite we define the product $AB$ in the usual way via $(AB)_{i,k} = \bigoplus_{j \in J}(A_{i,j}\otimes B_{j,k})$.  Analogously, it will sometimes also be convenient to view elements of $S^J$ as (row or column) `vectors' with entries indexed by $J$; for $\underline{x} \in S^J$ and $j \in J$ we write  $\underline{x}_j$ to denote the value computed by the function $\underline{x}$ on input $j$, and if $J$ is finite we define the product $A\underline{x}$ where $A \in S^{I \times J}$ in the obvious way. It is also straightforward to check that this multiplication is associative, that is, for all $A \in S^{I \times J}$, $B \in S^{J \times K}$, $C \in S^{K \times L}$, where $J$ and $L$ are both finite, we have $A(BC) = (AB)C$.

\begin{defi}
A weighted automaton $\mathcal{A}$ on a semiring $(S,\oplus,\otimes,0_S,1_S)$ over alphabet $\Sigma$ is defined by: a finite set of states $Q$; initial-state vector $\underline{\alpha} \in S^{Q}$ (which we view as a row vector); for all letters $a$ of $\Sigma$, transition matrices $M(a) \in S^{Q \times Q}$; and final-state vector $\underline{\eta} \in S^{Q}$ (which we view as a column vector). The automaton $\mathcal{A}$ computes a function $f_\mathcal{A} : \Sigma^* \rightarrow S$ via $f_\mathcal{A}(w_1 \cdots w_k) = \underline{\alpha} M(w_1)\cdots M(w_k)\underline{\eta}$. We say that $f_\mathcal{A} (w)$ is the value computed by $\mathcal{A}$ on input $w$. 
\end{defi}

A weighted automaton can be seen as a (labelled, directed) graph with set of vertices $Q$, such that (i) for each state $p \in Q$ there is an initial weight $\underline{\alpha}_p$ (represented by an incoming arrow at state $p$ with label $\underline{\alpha}_p$) and a final weight $\underline{\eta}_p$ (represented by an outgoing arrow at state $p$ with label $\underline{\eta}_p$); and (ii) for each pair of states $p,q \in Q$ and each $a \in \Sigma$ there is a transition (labelled edge) $p \xrightarrow{a\, : \, m} q$ where $m=M(a)_{p,q}$ is called the weight of the transition (and if $m=0_S$, then we often omit this edge from the graph). A run from a state $p$ to a state $q$ is a (directed) path in the graph and the weight of a run from $p$ to $q$ is the product of $\underline{\alpha}_p$  by the product of the weights on the transitions (taken in the order traversed) followed by $\underline{\eta}_q$. The value computed on a word $w$ is the sum of all weights of runs labelled by $w$. A state $p$ such that $\underline{\alpha}_p$ (resp. $\underline{\eta}_p$) is not $0_S$ is referred to as an initial (resp. final) state. We use both the matrix and graph vocabulary in this paper, as convenient.

\paragraph*{\textbf{Notations.}}
Throughout the paper, $\Sigma$ denotes a finite alphabet and $S$ a semiring (which will not be assumed to be commutative unless explicitly stated). The lower case symbols $a,b,c$ will be exclusively used to denote letters of an alphabet, whilst the symbols $u,v,w$ will be exclusively used to denote words over the alphabet $\Sigma$. We use calligraphic upper-case symbols such as $\mathcal{A}, \mathcal{B}, \mathcal{H}$ to denote weighted automata, whilst the lower-case symbols $f,g,h$ are always used to denote functions from $\Sigma^*$ to $S$.  Greek letters such as $\alpha, \beta, \lambda, \mu, \rho, \eta$ will be used to denote elements of the semiring $S$, upper-case symbols $A, B, \ldots, F$ will be used to denote (possibly infinite) matrices over $S$, and underlined symbols (e.g. $\underline{x}$ or $\underline{\lambda}$) will denote (possibly infinite) vectors over $S$; in both cases the index sets used  vary but will be made explicit. Upper-case symbols $I, J,\ldots, Z$ will typically be used for sets;  frequently these will be (countable) subsets of $S^{\Sigma^*}$ or $\Sigma^*$, and moreover such subsets will frequently be used as indexing sets.

\subsubsection*{\textbf{Specific semirings and examples:}} Besides giving a general framework valid for any semiring, we  consider several specific well-studied semirings that arise frequently in application areas.
We will consider the semiring $(\mathbb{R},+,\times,0,1)$ and its restriction to non-negative reals $\mathbb{R}_{\geq 0}$, integers $\mathbb{Z}$, and non-negative integers $\mathbb{N}$. We also consider the semiring $(\mathbb{R}\cup\{-\infty\},\max,+,-\infty,0)$ and its restrictions to $\mathbb{Z}\cup\{-\infty\}$ and $\mathbb{N}\cup\{-\infty\}$, denoted by $\mathbb{R}_{\max}$,  $\mathbb{Z}_{\max}$ and $\mathbb{N}_{\max}$ and the Boolean semiring $\mathbb{B}$. Note that all these semirings are commutative. Finally, we will consider the non-commutative semiring with domain the finite subsets of $\Sigma^*$, with operations given by union (playing the role of addition) and concatenation ($XY = \{xy \mid x\in X, y \in Y\}$, playing the role of multiplication) and neutral elements $\emptyset$ and $\{\varepsilon\}$. This semiring will be denoted $\mathcal{P}_{\rm fin}(\Sigma^*)$. 

\begin{exa}
We consider $\Sigma=\{a,b\}$. Figure~\ref{fig:examples} depicts several weighted automata where the initial (resp. final) states with weight $1_S$ are indicated with an ingoing (resp. outgoing) arrow, the other states have initial or final weights $0_S$. 
All six examples will be used in Section~\ref{section:hierarchy} and the example in Figure~\ref{fig:automata3} will be used as a running example to illustrate the different notions we introduce.
The automaton in Figure~\ref{fig:automata3} is considered over $\mathbb{N}_{\max}$ and computes the number of $a$'s if the word starts with an $a$, the number of $b$'s if the word starts with a $b$, and $-\infty$ on the empty word. The automaton depicted in Figure~\ref{fig:automata1} will be considered over $\mathbb{R}_{\max}$ and computes the length of the longest block of consecutive $a$'s in a word. The automaton in Figure~\ref{fig:automata5} is considered over $\mathbb{N}$ and on input $w$ computes: $2^{|w|_a}$  if $w$ starts with an $a$; $2^{|w|_b}$ if $w$ starts with a $b$; and $0$ if $w$ is the empty word. The automaton from Figure~\ref{fig:automata4} is viewed over $\mathbb{R}_{\geq 0}$ and computes $2^n - 1$ on the words $a^n$ for all positive integers $n$, and $0$ on any other word. Finally the automata in Figures~\ref{fig:automata5bis} and~\ref{fig:automata4bis} are over $\mathcal{P}_{\rm fin}(\Sigma^*)$. The first one has outputs: $\emptyset$ on the empty word; $\{a^{|w|_a}\}$ on all words $w$ starting with an $a$; and $\{b^{|w|_b}\}$ on all words $w$ starting with a $b$. The second one outputs $\{a^{|w|_a},b^{|w|_b}\}$ on all input words $w$. It is worth mentioning here that several of the examples above could be viewed as examples over smaller semirings (e.g. Figure~\ref{fig:automata1} could be considered over the semiring $\mathbb{N}_{\rm max}$). However, as we shall see below, changing the semiring over which one views a given example necessarily changes the space of hypothesis automata that one can construct using membership and equivalence queries. 
\end{exa}

\begin{figure}[]
     \centering
     \begin{subfigure}[b]{0.45\linewidth}
         \centering
         \includegraphics[width=0.9\textwidth]{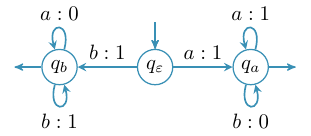}
         \caption{Over $\mathbb{N}_{\max}$: computes $|w|_x$ if  $w$ starts with letter $x\in \{a,b\}$ and $-\infty$ on $\varepsilon$.}
         \label{fig:automata3}
     \end{subfigure}
     \hfill
     \begin{subfigure}[b]{0.45\linewidth}
         \centering
         \includegraphics[width=0.9\textwidth]{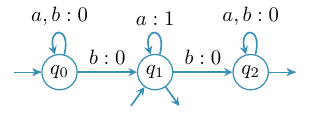}
         \caption{Over $\mathbb{R}_{\max}$: computes the length of the longest block of consecutive $a$'s in $w$.}
         \label{fig:automata1}
     \end{subfigure}

     \bigskip

     \begin{subfigure}[b]{0.45\linewidth}
         \centering
         \includegraphics[width=0.9\textwidth]{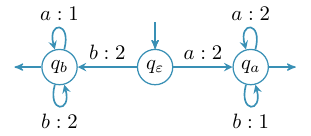}
         \caption{Over $\mathbb{N}$: computes $2^{|w|_x}$ if $w$ starts with letter $x \in \{a,b\}$ and $0$ if $w=\varepsilon$.}
         \label{fig:automata5}
     \end{subfigure}
     \hfill
    \begin{subfigure}[b]{0.45\linewidth}
         \centering
         \includegraphics[width=0.65\textwidth]{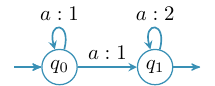}
         \caption{Over $\mathbb{R}_{\geq 0}$: computes $2^n - 1$ if $w=a^n$ for some $n \geq 1$, and $0$ otherwise.}
         \label{fig:automata4}
     \end{subfigure}

     \bigskip

     \begin{subfigure}[b]{0.45\linewidth}
         \centering
         \includegraphics[width=0.9\textwidth]{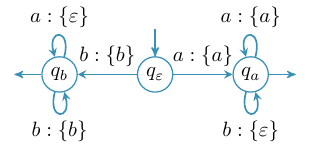}
         \caption{Over $\mathcal{P}_{\rm fin}(\Sigma^*)$: computes $\{x^{|w|_x}\}$ if $w$ starts with letter $x \in \{a,b\}$ and $\emptyset$ if $w=\varepsilon$.}
         \label{fig:automata5bis}
     \end{subfigure}
     \hfill
    \begin{subfigure}[b]{0.45\linewidth}
         \centering
         \includegraphics[width=0.64\textwidth]{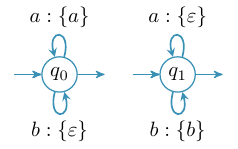}
         \caption{Over $\mathcal{P}_{\rm fin}(\Sigma^*)$: computes $\{a^{|w|_a}, b^{|w|_b}\}$}
         \label{fig:automata4bis}
     \end{subfigure}

        \caption{Example weighted automata and the values computed on input word $w$.}
        \label{fig:examples}
\end{figure}

\subsubsection*{\textbf{Linear combinations:}} Since we are going to consider non-commutative semirings, care needs to be taken when talking about linear combinations. Given a set $Y$, $S^Y$ is a left $S$-module via the following action: for each element $\underline{x}$ of $S^Y$ and $\lambda$ in $S$, $\lambda \otimes \underline{x}$ denotes the element of $S^Y$ computed from $\underline{x}$ by multiplying every component of $\underline{x}$ by $\lambda$ on the left. Given a subset $X \subseteq S^Y$, a left-linear combination over $X$ is one of the form $\bigoplus_{\underline{x} \in X} \lambda_{\underline{x}}\otimes \underline{x}$ for some $\lambda_{\underline{x}}$ in $S$, with only finitely many $\lambda_{\underline{x}}$ different from $0_S$. The left-semimodule generated by $X$ is then the subset of $S^Y$ containing those elements that can be written as a left-linear combination over $X$. An element of the left-semimodule generated by $X$ will simply be said to be left-generated by $X$. (Dually, one can view $S^Y$ as a right $S$-module, and make the corresponding definitions of right-linear combination, right-semimodule generated by $X$ etc.)

\subsubsection*{\textbf{Mirror:}} Finally, we will use the mirror operation on functions $\Sigma^* \to S$. The mirror of a word $w = w_1 w_2\cdots w_n$ where for all $i=1,\ldots, n$, $w_i \in \Sigma$ is the word $\bar{w} = w_n \cdots w_2 w_1$. For $f:\Sigma^* \to S$, the mirror of $f$, denoted by $\bar{f}:\Sigma^* \to S$ is defined as $\bar{f}(w) = f(\bar{w})$. Given a weighted automaton $\mathcal{A}$, the mirror of $\mathcal{A}$ is obtained by reversing each of its transitions and swapping the initial-state vector and the final-state vector. For example, the automaton in Figure 1B is equal to its mirror. Note that if $S$ is commutative and $f$ is computed by a weighted automaton $\mathcal{A}$ over $S$ then $\bar{f}$ is computed by the mirror of $\mathcal{A}$. However, there are cases when $S$ is non-commutative where a function can be computed by a weighted automaton over $S$ but its mirror cannot (for example, the function $w \mapsto \{w\}$ over $\mathcal{P}_{\rm fin}(\Sigma^*)$ is computed by a weighted automaton over $\mathcal{P}_{\rm fin}(\Sigma^*)$ but the mirror of this function ($w \mapsto \{\bar{w}\}$) is not).

\subsection{Learning with membership and equivalence queries}

We are investigating one type of learning done by membership and equivalence queries~\cite{A87,BV96,BM15} where the learner makes guesses by solving systems of left-linear (or right-linear) equations over $S$. Specifically, suppose that there is an oracle who can answer two types of questions about a function $f: \Sigma^* \rightarrow S$:
\begin{itemize}
	\item \textbf{Membership:} when provided with a word $w$: what is the value of $f$ on $w$? 
	\item \textbf{Equivalence:} when provided with a weighted automaton $\mathcal{H}$: is $f$ the function computed by $\mathcal{H}$? If the oracle answers ``no'', it will also give an input word $z$ on which $\mathcal{H}$ performs incorrectly.
\end{itemize}
A `membership' query is so-named since in the case of the Boolean semiring, this question amounts to asking whether a given word belongs to the rational language recognised by the automaton computing $f$. More generally, a ``membership query'' (sometimes also called ``output query'') provides a mechanism for a learner to gather data points $f(w)$ from which they may attempt to build an hypothesis automaton $\mathcal{H}$ consistent with these data points. An equivalence query provides a mechanism for a learner to test their hypothesis and receive useful feedback, in the form of success or counterexample. Throughout this paper we consider hypothesis automata constructed in a prescribed manner from a partial Hankel matrix, as we shall explain below.  In certain circumstances (i.e. by placing conditions on the weighted automaton and/or the semiring) it is already known that one can efficiently learn a weighted automaton computing the function $f$ using finitely membership and equivalence queries. Indeed, the underlying ideas introduced in the remainder of this section have been introduced and studied in the literature;  we give proofs of some known results for completeness, putting these into our general framework. Our presentation differs slightly and leads to questions that are answered in Sections~\ref{section:guessable} and~\ref{section:hierarchy} and represent our main contribution.

\subsubsection*{\textbf{Hankel matrix:}} As before, 
 we use the notation $\underline{x}$ for elements of $S^{\Sigma^*}$, and for a word $w \in \Sigma^*$ we write $\underline{x}_w$ to denote the $w$-th component of $\underline{x}$. The restriction of $\underline{x}$ to a subset $Z \subseteq \Sigma^*$ will be denoted by $\underline{x}_Z$. We shall also make use of the following right action of $\Sigma^*$ on $S^{\Sigma^*}$: given $\underline{x} \in S^{\Sigma^*}$ and $u \in \Sigma^*$, we denote by $\underline{x}\cdot u$ the element of $S^{\Sigma^*}$ defined as the shift of $\underline{x}$ by $u$, that is, $(\underline{x}\cdot u)_w = \underline{x}_{uw}$ for all words $w$. The right action of $\Sigma$ on $S^{\Sigma^*}$ commutes with the left action of $S$ on $S^{\Sigma^*}$ in the sense that for all $\lambda \in S$, $\underline{x} \in S^{\Sigma^*}$ and $u \in \Sigma^*$ we have $(\lambda \cdot \underline{x}) \cdot u = \lambda \cdot (\underline{x} \cdot u)$. Given a function $f: \Sigma^* \to S$, the infinite matrix $F$ whose rows and columns are indexed by $\Sigma^*$ and with entries $F_{w,w'} = f(ww')$ is called the Hankel matrix of $f$.  It will also be convenient to have a notation for elements of $S^{\Sigma^*}$ that arise as rows of the Hankel matrix. If $v \in \Sigma^*$ is a word and $f$ and $F$ are clear from context, we shall write $\langle v \rangle$ to denote the infinite row indexed by $v$ in $F$, viewed as an element of $S^{\Sigma^*}$. Note that if $v$ is a word, $f$ a function and $F$ its Hankel matrix, then the right action of $\Sigma^*$ on $\langle v \rangle$ corresponds to the obvious right action of $\Sigma^*$ on $v$, that is: $(\langle v \rangle \cdot u)_w = (\langle v \rangle)_{uw}  = F_{v,uw} = f(vuw) = F_{vu,w} =\langle vu \rangle_{w}$. 

\begin{exa}
Consider the function $f$ computed by our running example from Figure~\ref{fig:automata3}. Its Hankel matrix is given by $F_{\varepsilon, \varepsilon} = -\infty$, $F_{\varepsilon, aw} = |aw|_a$, $F_{\varepsilon, bw} = |bw|_b$, $F_{aw, w'} = |aww'|_a$ and $F_{bw, w'} = |bww'|_b$ for all words $w,w'$.    
\end{exa}

\subsubsection*{\textbf{Closed sets and generating sets:}}
Let  $Z \subseteq \Sigma^*$ and $f: \Sigma^* \rightarrow S$. A subset $X \subseteq S^{\Sigma^*}$ is said to be: (i) left-closed if for all $\underline{x}$ in $X$ and $a \in \Sigma$, $\underline{x}\cdot a$ is left-generated by $X$; (ii) left-closed on $Z$ if for all $\underline{x}$ in $X$ and $a \in \Sigma$, $(\underline{x}\cdot a)_Z$ is left-generated by the elements of $X$ restricted to $Z$; (iii) row-generating for $f$ if every row of the Hankel matrix of $f$ is left-generated by $X$. By an abuse of language, we say that a set of words\footnote{Using $Q$ as the label for a set of words might seem odd at this point, but the reason behind this will become apparent very shortly.} $Q$ is left-closed (resp. left-closed on $Z$, row-generating for $f$) if the set $\underline{Q}=\{ \langle q \rangle: q \in Q\}$ has this property.

The following simple observations are key in what follows.
\begin{lem}
\label{lem:leftgen}
Let $X \subseteq S^{\Sigma^*}$.
\begin{enumerate}
    \item If $X$ is left-closed and left-generates the row $\langle \varepsilon \rangle$ of the Hankel matrix of $f$ then $X$ is row-generating for $f$.
    \item If $X$ is a subset of the rows of the Hankel matrix and is a row-generating set for $f$, then $X$ is left-closed (and left-generates the row $\langle \varepsilon \rangle$).
\end{enumerate}
\end{lem}
\begin{proof}
(1) Since $X$ left generates  $\langle \varepsilon \rangle$ we may write $\langle \varepsilon \rangle = \bigoplus_{\underline{x}  \in X} \lambda_{\underline{x}} \otimes \underline{x}$, where only finitely many of the coefficients $\lambda_{\underline{x}}$ are non-zero. Then for all words $u,v \in \Sigma^*$ we have $$\langle u \rangle_v = \langle \varepsilon \rangle_{uv} = \bigoplus_{\underline{x}  \in X} \lambda_{\underline{x}} \otimes \underline{x}_{uv} = \bigoplus_{\underline{x}  \in X} \lambda_{\underline{x}} \otimes (\underline{x} \cdot u)_v.$$
Then, using the fact that $X$ is left-closed, an easy induction on the length of $u$ then gives that $\underline{x}\cdot u$ is left-generated by $X$, in turn showing that each row  $\langle u \rangle$ of the Hankel matrix is a left-linear combination of elements of $X$.\\ \ 
(2) That $X$ left-generates $\langle \varepsilon \rangle$ is completely clear from definition of row generating set. Since $X$ is a subset of the rows of the Hankel matrix of $f$, each element of $X$ may be written in the form $\langle w \rangle$ for some $w \in \Sigma^*$, in which case $wa \in \Sigma^*$ for all $a \in \Sigma$, and hence $\langle w \rangle\cdot a = \langle wa \rangle$ is left-generated by $X$.
\end{proof}

\begin{exa}
Consider again the function $f$ computed by our running example from Figure~\ref{fig:automata3} over the semiring $\mathbb{N}_{\rm max}$ whose multiplicative operation $\otimes$ is given by usual addition. The singleton set $\{\langle \varepsilon \rangle\}$ is not left-closed: indeed, $\langle a \rangle$ is not left-generated by this element since $\langle \varepsilon \rangle_\varepsilon = -\infty$ but $\langle a \rangle_\varepsilon = 1$, and hence there is no left-linear combination. The two element set $\{\langle \varepsilon \rangle, \langle a \rangle\}$ is left-closed  on $\{\varepsilon\}$, since $\langle b \rangle_\varepsilon =  \langle ab \rangle_\varepsilon = 1 = 0+1 = 0 \otimes \langle a \rangle_\varepsilon$ and $\langle aa \rangle_\varepsilon = 2 = 1+1 = 1 \otimes \langle a \rangle_\varepsilon$. The set $\{\langle \varepsilon \rangle, \langle a \rangle, \langle b \rangle \}$ is left-closed  and also row-generating since for $x \in \{a,b\}$, $\langle xw \rangle = |w|_x \otimes \langle x \rangle$.
\end{exa}

It turns out that if one can find a finite  subset  $X \subseteq S^{\Sigma^*}$ that is left-closed and such that the row $\langle \varepsilon \rangle$ is left-generated by $X$, then by solving the corresponding system of $S$-linear equations involving the (infinite) vectors in $X \cup \{\langle \varepsilon\rangle\}$ one can construct an automaton computing the correct function. We call this a Hankel automaton and describe its construction below. However, in practice, using finitely many membership queries only a finite submatrix of the Hankel matrix can be obtained. An hypothesis automaton is defined in the same way as a Hankel automaton, replacing the finite subset $X$ by a finite set of \emph{rows of the Hankel matrix} (i.e. indexed by a set of words) \emph{restricted to a finite set of columns}. For Boolean automata, it is always possible to find a finite set amongst the rows of the Hankel matrix and a finite set of columns such that the corresponding hypothesis automaton computes the target function. The $L^*$-algorithm is based on this idea. 

\subsubsection*{\textbf{Hankel automata and hypothesis automata:}}
We first set up some notation to describe the construction of hypothesis automaton.  
\begin{nota}
	\label{notation:lambda}
Given a function $f$ (and in particular the row indexed by $\varepsilon$ in its Hankel matrix), a finite non-empty subset $Q$ of vectors of $S^{\Sigma^*}$ and a non-empty subset $T$ of $\Sigma^*$, we define the left solution set $\Lambda_{Q, T}$  to be the set of vectors $\underline{\lambda} \in S^{Q \cup (Q \times \Sigma \times Q)}$ satisfying:
\begin{itemize}
\item $\langle \varepsilon \rangle_T = \bigoplus_{\underline{q} \in Q} \underline{\lambda}_{\underline{q}} \otimes \underline{q}_T$, and
\item for all $\underline{q} \in Q$ and all $a \in \Sigma$, $(\underline{q}\cdot a)_T= \bigoplus_{\underline{p} \in Q} \underline{\lambda}_{\underline{q},a,\underline{p}} \otimes \underline{p}_T$.
\end{itemize}
If $T=\Sigma^*$,  we will simply write $\Lambda_{Q}$ instead of $\Lambda_{Q, \Sigma^*}$.
Note that the left solution sets $\Lambda_{Q,T}$ and $\Lambda_Q$ depend on the function $f$ (and hence also the semiring $S$), although we have chosen to suppress this dependence in our notation for readability. 
If $W$ is a finite subset of $\Sigma^*$ we will also, by a slight abuse of notation, write $\Lambda_W$
and $\Lambda_{W,T}$
to mean $\Lambda_{Q}$ and $\Lambda_{Q,T}$ where 
$Q = \{\langle w \rangle: w \in W\}$.
\end{nota}

\begin{rem}
\label{rem:lambda}
An element $\underline{\lambda} \in \Lambda_{Q,T}$ represents a collection of coefficients witnessing that $\langle \varepsilon \rangle_T$ and each of the $(\underline{q} \cdot a)_T$ (for all $\underline{q} \in Q$ and all $a \in \Sigma$) is a left-linear combinations of the $\underline{q}_T$ for $q$ in $Q$. In other words, $\Lambda_{Q, T}$ represents the solution space to the system of $S$-linear equations outlined in Notation \ref{notation:lambda}.
The key idea here is that the set $\Lambda_{Q, T}$ is non-empty if and only if a solution to this system exists; in the language introduced above, if and only if $Q$ is left-closed on $T$ and $\langle \varepsilon \rangle_T$ is left-generated by $Q$ restricted to $T$.
\end{rem}
\begin{rem}
\label{rem:subsets}
Let $T \subseteq \Sigma^*$, and let $Q \subseteq S^{\Sigma^*}$ be a finite set. If $T \subseteq T’ \subseteq \Sigma^*$, then it is clear from the above equations that $\Lambda_{Q, T} \supseteq \Lambda_{Q, T’}\supseteq \Lambda_{Q}$. However, if $Q'$ is a finite set contained in or contained by $Q$, then there is no containment between $\Lambda_{Q, T}$ (which lies in $S^{Q \cup (Q \times \Sigma \times Q)}$) and $\Lambda_{Q', T}$ (which lies in $S^{Q' \cup (Q' \times \Sigma \times Q')}$).
\end{rem}
\begin{lem}
\label{lem:Qincrement}
Let $f:\Sigma^* \rightarrow S$ be a function, $T \subseteq \Sigma^*$, and suppose that $Q' \subseteq Q$ are finite sets of rows of the Hankel matrix of $f$. If $\Lambda_{Q'} \neq \emptyset$ then $\Lambda_{Q, T} \neq \emptyset$.
\end{lem}
\begin{proof}
Since each element of $Q$ is a row of the Hankel matrix, there exists a finite set of words $W \subseteq \Sigma^*$ such that $Q = \{\langle w \rangle: w \in W\}$. Since $\Lambda_{Q'} \neq \emptyset$, we have that $Q'$ is left-closed and left-generates the row $\langle \varepsilon \rangle$, and so $Q'$ is row-generating for $f$, by Lemma \ref{lem:leftgen}. In particular, $\langle \varepsilon \rangle$ and each row $\langle wa \rangle$ where $w \in W$ and $a \in \Sigma$ can be expressed as a left linear combination of the elements of $Q'$, and hence also of $Q$. Since $\langle wa \rangle = \langle w \rangle \cdot a$, it then follows easily that the system of equations given in Notation \ref{notation:lambda} has a solution.
\end{proof}

\begin{defi}
\label{definition:hankel}
Let $f: \Sigma^* \rightarrow S$ be a function, and suppose that $Q$ is a finite subset of $S^{\Sigma^*}$ and $T$ a subset of $\Sigma^*$ such that $\Lambda_{Q, T} \neq \emptyset$. Then for each $\underline{\lambda} \in \Lambda_{Q, T}$, we define a finite weighted automaton $\mathcal{H}_{Q, T, \underline{\lambda}}$ with:
\begin{itemize}
\item finite set of states $Q$;
\item initial-state vector $\underline{\lambda}_Q$;
\item final-state vector $(\underline{q} _\varepsilon)_{\underline{q} \in Q}$; and
\item for all $\underline{q},\underline{p} \in Q$ and all $a \in \Sigma$, a transition from state $\underline{q}$ to state $\underline{p}$ labelled by $a$ with weight $\underline{\lambda}_{\underline{q},a,\underline{p}}$.
\end{itemize}
If $Q$ is a finite set of rows of the Hankel matrix of $f$ and $T = \Sigma^*$,  we say $\mathcal{H}_{Q, T, \underline{\lambda}}$ is a \emph{Hankel automaton} and write simply $\mathcal{H}_{Q, \underline{\lambda}}$ to reduce notation. If $Q$ is a finite set of rows of the Hankel matrix and $T$ is finite, we say that $\mathcal{H}_{Q, T, \underline{\lambda}}$ is an \emph{hypothesis automaton}; in this case we may view $Q$ as a finite set of words by identifying with an appropriate index set (note: there need not be a unique set of words).\end{defi}

 \begin{exa}
Consider the function $f$ over $\mathbb{N}_{\max}$ computed by our running example from Figure~\ref{fig:automata3}. The set $\{ \langle \varepsilon \rangle, \langle a \rangle, \langle b \rangle\}$ containing three rows of the Hankel matrix clearly left-generates $\langle \varepsilon \rangle$ and, as shown previously, is left-closed. Let $q_\varepsilon, q_a$ and $q_b$ denote these rows (omitting the underlining of these elements for ease of reading). Then considering the coefficients $\underline{\lambda}_{q_\varepsilon}=0$, $\underline{\lambda}_{q_\varepsilon,a,q_a} = \underline{\lambda}_{q_\varepsilon,b,q_b} = \underline{\lambda}_{q_a,b,q_a} = \underline{\lambda}_{q_b,a,q_b} = 0$, $\underline{\lambda}_{q_a,a, q_a} = \underline{\lambda}_{q_b,b,q_b} = 1$ and all non-specified coefficients $-\infty$, the corresponding Hankel automaton would be the automaton with three states obtained from the one in Figure~\ref{fig:automata3} by replacing the weight of the transitions starting at $q_{\varepsilon}$ by $0$, and replacing the final weights by $1$. For $Q = \{q_\varepsilon, q_ a\}$ and $T = \{\varepsilon\}$ we have also seen that $Q$ is left-closed on $T$. By taking $\underline{\lambda}_{q_\varepsilon}=0$, $\underline{\lambda}_{q_\varepsilon,a,q_a} = \underline{\lambda}_{q_\varepsilon,b,q_a} = \underline{\lambda}_{q_a,b,q_a} = 0$, $\underline{\lambda}_{q_a,a, q_a} = 1$ and all remaining coefficients $-\infty$, one obtains an hypothesis automaton with two states. Notice that this automaton does not compute $f$; it computes $|w|_a$ if $w$ starts with an $a$, $|w|_a + 1$ if $w$ starts with a $b$, and $-\infty$ if $w =\varepsilon$. Note that if $ \{\varepsilon, aw\} \subseteq T'$ for some word $w$, then $Q = \{q_\varepsilon, q_ a\}$ is not left-closed on $T'$. 
\end{exa}

Fliess's theorem \cite{F74} states that a function is computed by a weighted automaton over the field of reals if and only if its Hankel matrix has finite rank, and moreover, this rank is exactly the minimal number of states of a weighted automaton computing the function. There are many equivalent definitions of the rank of a matrix with entries over a field whose analogues cease to coincide over an arbitrary semirings, however, the proof of Fliess' Theorem can be easily adapted to obtain the following statement, and since the details are instructive, we include a short proof.
\begin{thm}
\label{theorem:prelim}
Let $f: \Sigma^* \rightarrow S$ be a function. The following are equivalent:
\begin{enumerate}
\item there exists a finite subset $Q \subseteq S^{\Sigma^*}$ that is both left-closed and left-generates the row indexed by $\varepsilon$ in the Hankel matrix of $f$ (i.e. $\Lambda_Q \neq \emptyset$ for some finite subset $Q\subseteq S^{\Sigma^*}$);
\item $f$ is computed by a finite-state weighted automaton $\mathcal{A}$ over $S$.   
\end{enumerate}
Specifically, if $\Lambda_Q \neq \emptyset$, then $\mathcal{H}_{Q, \underline{\lambda}}$ computes $f$ for all $\underline{\lambda} \in \Lambda_Q \neq \emptyset$, whilst if $f$ is computed by the finite-state weighted automaton $\mathcal{A}$ over $S$ with state set $Q$, then $\Lambda_{\underline{Q}} \neq \emptyset$ for $\underline{Q} = \{\underline{q}: q \in Q\} \subseteq S^{\Sigma^*}$ where for each $w \in \Sigma^*$ we define  $\underline{q}_w$ to be the value on input $w$ computed by the automaton $\mathcal{A}_q$ obtained from $\mathcal{A}$ by making $q$ the unique initial state with weight $1_S$ and all other weights are the same.
\end{thm}

\begin{rem}
	\label{rem:rows}
Before presenting the proof, we alert the reader to the fact that in Theorem \ref{theorem:prelim} the set $Q$ is not assumed to be a set of rows of the Hankel matrix of $f$, and as we shall see below, in general, a function computed by a finite-state weighted automaton over an arbitrary semiring may or may not satisfy $\Lambda_Q \neq \emptyset$ for some finite subset $Q$ of the rows of its Hankel matrix.  However, in view of Notation \ref{notation:lambda}, Remark \ref{rem:lambda} and Definition \ref{definition:hankel}, \textbf{if there exists} a finite set $Q$ of rows of the Hankel matrix such that $\Lambda_Q$ is non-empty, then  for any finite set $T \subseteq \Sigma^*$ and any $\underline{\lambda} \in \Lambda_Q$ we have $\underline{\lambda} \in \Lambda_Q \subseteq \Lambda_{Q, T}$, so that (by construction) the hypothesis automaton $\mathcal{H}_{Q, T, \underline{\lambda}}$ constructed from the finite $(Q,T)$ submatrix of the Hankel matrix of $f$, will be equal to the Hankel automaton $\mathcal{H}_{Q, \underline{\lambda}}$ and, so by Theorem \ref{theorem:prelim}, will compute $f$. The success of Angluin-type algorithms in the existing literature is therefore typically based on the ability to find (using finitely many membership and equivalence queries): a finite set of words $Q$ such that $\Lambda_Q$ is non-empty, together with a finite set of words $T$ such that $\Lambda_{Q, T}=\Lambda_Q$ in the given setting.
\end{rem}	
\begin{proof}
Suppose that $Q \subseteq S^{\Sigma^*}$ is a finite subset such that $\Lambda_Q \neq \emptyset$ (or equivalently by Remark \ref{rem:lambda}, $Q$ is left-closed and left-generates $\langle \varepsilon \rangle$) and let $\underline{\lambda} \in \Lambda_{Q}$. We shall show that $\mathcal{H}_{Q, \underline{\lambda}}$ computes $f$. To begin with, for each $\underline{q}$ in $Q$, let $\mathcal{H}_{\underline{q}}$ be the automaton obtained from $\mathcal{H}_{Q,\underline{\lambda}}$ by imposing the condition that $\underline{q}$ is the unique initial state with weight $1_S$, and keeping all other weights the same. We claim that for all
$w$ in $\Sigma^*$, $\underline{q}_w$ is equal to the weight computed by $\mathcal{H}_{\underline{q}}$ on input $w$. The proof is by induction on the length of $w$. If $w=\varepsilon$, then $\underline{q}_\varepsilon$ is (by construction) equal to the final weight of $\underline{q}$ in $\mathcal{H}_{Q,\underline{\lambda}}$ which proves the property. For $a$ in $\Sigma$ and $w$ in $\Sigma^*$, $\underline{q}_{aw} = (\underline{q}\cdot a)_w = \bigoplus_{\underline{p} \in Q}  \underline{\lambda}_{\underline{q},a,\underline{p}} \otimes \underline{p}_w$ since $\underline{\lambda} \in \Lambda_Q$. By induction hypothesis, $\underline{p}_w$ is the weight computed by $\mathcal{H}_{\underline{p}}$ on input $w$, and moreover $\underline{\lambda}_{\underline{q},a,\underline{p}}$ is (by construction)  the weight of the transition from state $\underline{q}$ to state $\underline{p}$ labelled by $a$. So $\underline{q}_{aw}$ is indeed the weight computed by $\mathcal{H}_{\underline{q}}$ on input $aw$. To conclude the proof, recall that $\langle \varepsilon \rangle$ denotes the row of $\varepsilon$ in the Hankel matrix of $f$, and so by definition of the Hankel matrix of $f$ and the fact that $\underline{\lambda} \in \Lambda_Q$ we have $f(w) = \langle \varepsilon \rangle_w = \bigoplus_{\underline{q} \in Q} \underline{\lambda}_q \otimes \underline{q}_w$; by our previous observation (together with the fact that $\underline{\lambda}_Q$ is the initial state vector of $\mathcal{H}_{Q, \underline{\lambda}}$), this is then exactly the value computed by $\mathcal{H}_{Q, \lambda}$ on input $w$.
	
	Now suppose that $f:\Sigma^* \to S$ is computed by a finite-state weighted automaton $\mathcal{A}$ over $S$ with set of states $Q$. For each $q \in Q$, let $_q\mathcal{A}$ (resp. $\mathcal{A}_q$) denote the automaton obtained from $\mathcal{A}$ by making $q$ the unique initial (resp. final) state with weight $1_S$. For each $q \in Q$, define $\underline{q} \in S^{\Sigma^*}$ where for all $u \in \Sigma^*$, $\underline{q}_u$ is equal to the value computed by $_q\mathcal{A}$ on input $u$. Let $\underline{Q} = \{\underline{q} \mid q\in Q\}$. First $\underline{Q}$ is left-closed since, by definition, for all words $u\in \Sigma^*$, all states $q \in Q$ and all $a\in \Sigma$ we have
 $$(\underline{q} \cdot a)_u = \underline{q}_{au} = \bigoplus_{p \in Q} (\nu_{p} \otimes \underline{p}_u)$$ where $\nu_p$ is the weight of the transition labelled by $a$ in $\mathcal{A}$ from $q$ to $p$.
	Let $w$ be a word in $\Sigma^*$ and as before let $\langle w \rangle$ denote the infinite row of the Hankel matrix of $f$ indexed by $w$. Let $\mu_{q}(w)$ be the value computed by $\mathcal{A}_q$ on input $w$. Then, by definition, $\langle w \rangle = \bigoplus_{q \in Q} \mu_{q}(w) \otimes \underline{q}$, showing that $\underline{Q}$ is row-generating for $f$ and hence in particular $\langle \varepsilon \rangle$ is left-generated by $\underline{Q}$.
\end{proof}

\begin{rem}
	\label{rem:dual}
	In this section we have taken an approach based on ``rows'', describing how to construct hypothesis automata utilising finitely many partial rows of the Hankel matrix. Clearly, a dual construction can be made working with the transpose of the Hankel matrix, that is, using  partial columns and corresponding notions of right-linear combinations, right semimodules, right-generating sets,  right-closed sets, column-generating sets, etc.  In the text that follows we shall refer to the analogues of Hankel automata and hypothesis automata constructed by this dual method as co-Hankel automata and co-hypothesis automata. For the readers convenience we outline the details of these constructions and summarise the dual definition to those presented in the main text in Appendix \ref{section:swap}. 
\end{rem}

The previous result together with its dual (see Theorem~\ref{theorem:comain} of the appendix) give the following:
\begin{cor}
Let $S$ be a semiring and $f: \Sigma^* \rightarrow S$. The following are equivalent:
\begin{enumerate}
	\item  $f$ is computed by a finite-state weighted automaton $\mathcal{A}$ over $S$;
	\item the rows of the Hankel matrix of $f$ lie in a left-subsemimodule of $S^{\Sigma^*}$ generated by a finite left-closed set;
	\item the columns of the Hankel matrix of $f$ lie in a right-subsemimodule of $S^{\Sigma^*}$ generated by a finite right-closed set.
\end{enumerate}
\end{cor}	

\begin{rem}
\label{rem:fields1}
Notice that if $S$ is a field (as in \cite{F74}) then the second condition of the previous result implies that the rows of the Hankel matrix must lie in a finitely generated vector space, which in turn forces the vector space spanned by the rows of the Hankel matrix to be finitely generated. Thus if $f$ is a function computed by a finite automaton over a field $S$,  one can find finite sets of \emph{words} $Q,P \subseteq \Sigma^*$ such that the rows of the Hankel matrix indexed by $Q$ are row-generating and the columns of the Hankel matrix indexed by $P$ are column-generating. However, over an arbitrary semiring, a subsemimodule of a finitely generated semimodule need not be finitely generated, and so there is no reason to suspect that the left-semimodule generated by the rows of the Hankel matrix is finitely generated. Moreover, it can happen that the left-semimodule generated by the rows of the Hankel matrix is finitely generated, whilst the right-semimodule generated by the columns is not (or vice versa).
\end{rem}

\subsubsection*{\textbf{Arising questions and structure of the paper:}}

It is known that for all functions computed by weighted automata over fields there is an efficient algorithm (using linear algebra) to find, by membership and equivalence queries, finite
sets of \textbf{words}
$Q$ and $T$ such that
$\Lambda_{Q,T}$ is non-empty
together with an  element 
$\underline{\lambda}
\in \Lambda_{Q,T}$  such that the hypothesis automaton $\mathcal{H}_{Q,T,\underline{\lambda}}$ computes the target function (see for example \cite{BV96}). For an arbitrary (but fixed) semiring $S$, the following questions arise:
\begin{itemize}
\item For which functions $f: \Sigma^* \rightarrow S$ does there \emph{exist} an hypothesis automaton $\mathcal{H}_{Q, T, \underline{\lambda}}$ computing $f$? We will call the class of functions which can be computed by \emph{at least one} hypothesis automaton (which recall is built using a finite portion of the Hankel matrix) weakly guessable. Clearly, if $f$ is not weakly guessable then any learning algorithm based around making hypotheses of the kind under discussion (constructed from certain left linear combinations of partial rows of the Hankel matrix) will fail to learn $f$. 
\item For weakly guessable functions over $S$, is there a learning algorithm using membership and equivalence queries that allows a learner to construct a finite sequence of hypothesis automata, where the last in the sequence computes the target function?
In other words: when there exists at least one correct hypothesis, can we devise a strategy to find one? A possible issue here is that even when we have settled upon a suitable set $Q$ with $\Lambda_Q \neq \emptyset$, it may be that there is no finite set $T$ such that $\Lambda_{Q, T}= \Lambda_Q$, and therefore no obvious general method to discover a ``correct'' $\underline{\lambda} \in \Lambda_Q$  in finite time. We say that $f$ is guessable if there exist finite sets $Q$ and $T$ such that $\Lambda_{Q,T} = \Lambda_Q \neq \emptyset$ and so \emph{every} hypothesis automaton $\mathcal{H}_{Q, T, \lambda}$ with $\lambda \in \Lambda_{Q,T}$ computes $f$. Since $\Sigma^*$ is countable,
so too is the set of all pairs of \emph{finite} subsets of $\Sigma^*$; thus if armed with an enumeration of all such pairs together with the ability to compute an element of every left solution set $\Lambda_{Q,T}$  (or else determine that the set is empty), it is clear that there is a naive algorithm (working through the list of all pairs of subsets, and performing an equivalence query where possible) that
is guaranteed to succeed on the 
class of guessable functions (since $\Lambda_{Q,T}=\Lambda_Q \neq \emptyset$ holds for at least one pair).
\item In order to find a correct hypothesis automaton, some of the existing approaches and algorithms typically alternate between `incrementation' of the sets of words $Q$ and $T$ adding elements to $Q$ in a prescribed manner when the left-solution set $\Lambda_{Q,T}$ is found to be empty, and adding elements to $T$ in a prescribed manner when a counter example is given, eventually converging on a pair for which $\Lambda_{Q, T} = \Lambda_Q \neq \emptyset$. A potential problem here is that, in general, one does not know for a given set $Q'$  whether it is worth persevering by continuing to add words to $T$ or not (in other words, even though there exists a finite set $Q$ such that $\Lambda_{Q,T} = \Lambda_Q \neq \emptyset$ for some finite set $T$, it may be the case that for the current choice $Q'$ we have $\Lambda_{Q', T} \supsetneq \Lambda_{Q'}$  for all finite sets $T$). With this in mind, we say that a guessable function $f$ is strongly guessable if for all finite sets of words $Q$ there exists a finite set of words $T$, such that $\Lambda_{Q,T}=\Lambda_Q$. For strongly guessable functions, the general problem therefore reduces to the question of how to find a set $Q$ with the property that $\Lambda_{Q} \neq \emptyset$.
\item Finally, given a partial Hankel matrix one could ask the corresponding questions concerning existence of co-hypothesis automata constructed from right-linear combinations of partial columns of the Hankel matrix (see Appendix \ref{section:swap} for details), leading to the dual notion of weakly co-guessable, co-guessable, and strongly co-guessable functions.
\end{itemize}
\noindent 
In Section~\ref{section:guessable}, we define the classes of weakly guessable, guessable and strongly guessable functions and give several equivalent characterisations for them in Theorem~\ref{prop:weakly} and \ref{prop:guessable}. In Section~\ref{section:hierarchy}, we give a general hierarchy for these classes with some specific  examples, and give conditions for this hierarchy to collapse. This allows us to have a clearer picture of what happens for several well-studied examples of semirings arising in application areas, where limitations of this approach may be of interest.  This opens many questions for further research that are stated in the concluding Section~\ref{section:conclusion}. 

\section{Guessable functions}
\label{section:guessable}

In this section, we define weakly guessable,  guessable, and strongly guessable functions and give characterisations for them.

\subsection{Weakly guessable functions}

\begin{defi}
\label{definition:weakly}
A function $f: \Sigma^* \rightarrow S$ is said to be weakly guessable if there exists a finite set $Q$ of rows of the Hankel matrix of $f$ such that $\Lambda_Q$ is non-empty. 
\end{defi}

Observe that by Theorem~\ref{theorem:prelim}, this implies that for all $\underline{\lambda} \in \Lambda_Q$ the Hankel automaton $\mathcal{H}_{Q,\underline{\lambda}}$ computes $f$ and hence for any non-empty set of words $T$, there exists an hypothesis automaton $\mathcal{H}_{Q,T,\underline{\lambda}}$ computing $f$ since (by Remark \ref{rem:subsets}) $\Lambda_Q \subseteq \Lambda_{Q,T}$. So this definition means that there is ``some hope'' that an algorithm \`a la Angluin can find a correct automaton, but does not really give a handle on what these functions look like. The following definition (literal automaton) turns out to provide a syntactic characterisation of these functions in terms of automata, and links to residual automata. A literal automaton is a residual automaton, as defined for example for non-deterministic or nominal automata~\cite{GGV20,MS22} with extra conditions to deal with the weights. It is very close to the reduced residual automata in the stochastic settings defined in~\cite{DE08}, with a slight difference in the way the weights are handled - the reason being that, while non-zero probabilities have an inverse for the multiplication, this is not systematically the case in general semirings. We use the terminology literal automata to avoid any confusion with~\cite{DE08}, our choice of terminology reflecting the fact that each state of a literal automaton can be labelled by a word with the property that any run from the initial state labelled by that word terminates in the given state.

\begin{defi}
\label{definition:literal}
A weighted automaton over $S$ with set of states $Q$ is said to be literal if it has a unique initial state with weight $1_S$ and there exists a prefix-closed set of words $W \subseteq \Sigma^*$ and a bijection $\sigma: Q \rightarrow W$ such that for all $q \in Q$:
	\begin{itemize}
		\item there is a unique run, $\gamma_q$, starting at the initial state and labelled by $\sigma(q)$,
		\item each transition in $\gamma_q$ has weight $1_S$,
		\item $\gamma_q$ terminates at state $q$.
	\end{itemize}
\end{defi}
\begin{rem}
\label{rem:litrows}
In a literal automaton $\mathcal{A}$ with set of states $Q$, each state $q$ can be labelled by a word $\sigma(q)$ which when taken as input terminates in  state $q$ (via a unique run, which has weight $1_S$), and so in this case the set $\underline{Q} = \{\underline{q}  \mid q \text{ is a state of } \mathcal{A}\}$ defined in Theorem~\ref{theorem:prelim} is easily seen to be precisely the set of rows of the Hankel matrix of $\mathcal{A}$ indexed by the words chosen to label the states, that is, $\underline{Q} = \{\langle \sigma(q) \rangle: q \in Q\}$.
\end{rem}

Finally, we link the functions computed by literal automata to a weaker variant of a property given in~\cite{HKRS20}, demonstrating that this corresponds exactly to the weakly guessable functions.

\begin{defi}
\label{definition:wacc}
Given a function $f:\Sigma^* \to S$ and $F$ its Hankel matrix, $f$ satisfies the weak ascending chain condition if for all chains of left-semimodules  $X_0 \subseteq X_1 \subseteq X_2 \subseteq \ldots $ such that
\begin{itemize}
    \item each $X_i$ is generated by a finite set of rows of $F$;
    \item each row of $F$ belongs to some $X_i$
\end{itemize}
there exists $n$ such that for all $m \geq n$, $X_m = X_n$.
\end{defi}

\begin{rem}
\label{rem:Noetherian}
The terminology ``weak'' in the previous definition reflects the fact that we only consider \emph{certain chains} of left-semimodules. A left-semimodule is said to be Noetherian if \emph{all} chains of left-subsemimodules satisfy the ascending chain condition.  Thus, if every left-semimodule generated by a finite set of rows of $F$ is Noetherian, then $f$ satisfies the weak ascending chain condition above, but not conversely (in general, a submodule of a finitely generated semimodule need not be finitely generated, nor need a subsemimodule of a semimodule generated by rows of the Hankel matrix also be generated by rows of the Hankel matrix).	
\end{rem}

\begin{thm}
\label{prop:weakly}
Given a function $f:\Sigma^* \to S$, the following assertions are equivalent:
\begin{enumerate}
\item $f$ is weakly guessable;
\item there exists a literal automaton computing $f$;
\item there exists a finite set of rows of the Hankel matrix of $f$ that is row-generating for $f$;
\item $f$ satisfies the weak ascending chain condition.
\end{enumerate}  
\end{thm}

\begin{rem}
	\label{rem:rows2}
    Before presenting the proof, it is instructive to  compare Theorem \ref{prop:weakly} (characterising weakly guessable functions) with Theorem \ref{theorem:prelim} (characterising functions computed by a finite weighted automata). Notice that point (3) of Theorem \ref{prop:weakly} states that there is a finite subset of the \emph{rows of the Hankel matrix} that left-generates every row of the Hankel matrix, whilst Theorem \ref{theorem:prelim}, states that the rows of the Hankel matrix can be generated by finitely many elements of $S^{\Sigma^*}$ (which need not themselves be rows of the Hankel matrix; see also Remark \ref{rem:rows}).
\end{rem}
\begin{proof}
	(1. implies 2.) Since $f$ is weakly guessable there is a finite set of rows $Q$ of the Hankel matrix such that $\Lambda_Q$ is non-empty. Let $W$ be a finite set of words with $|W|=|Q|$ such that $Q$ is the set of rows of the Hankel matrix indexed by $W$ and let $W'$ denote the set of all prefixes of words in $W$,  with $Q'$ the corresponding set of rows of the Hankel matrix, so that $Q \subseteq Q' = \{\langle q \rangle \mid q \in W'\}$. Since $\Lambda_Q \neq \emptyset$, we have seen (in Lemma \ref{lem:Qincrement}) that $\Lambda_{Q'} \neq \emptyset$; we aim to show that there exists $\underline{\mu} \in \Lambda_{Q'}$ such that the corresponding Hankel automaton $\mathcal{H}_{Q', \underline{\mu}}$ (which by Theorem \ref{theorem:prelim} computes $f$) is literal. To this end, we  let $\underline{\lambda} \in \Lambda_Q$ so that (by Theorem ~\ref{theorem:prelim}) $\mathcal{H}_{Q,\underline{\lambda}}$ computes $f$ and use the values in $\underline{\lambda}$ to define $\underline{\mu} \in S^{Q' \cup (Q' \times \Sigma \times Q')}$ as follows:
	\begin{itemize}
		\item $\underline{\mu}_{\langle q \rangle} = 1_S$ if and only if $\langle q \rangle =\langle \varepsilon \rangle$ and $\underline{\mu}_{\langle q \rangle} = 0_S$ for all other  $\langle q \rangle  \in Q'$,
		\item if $q \in W'$, $a \in \Sigma$ and $qa \in W'$, then $\underline{\mu}_{\langle q \rangle, a ,\langle qa\rangle} = 1_S$, and $\underline{\mu}_{\langle q \rangle,a,\langle p\rangle} = 0_S$ for all other $\langle p \rangle \in Q'$;
		\item if $q  \in W$, $a \in \Sigma$ and $qa$ is not in $W'$, then $\underline{\mu}_{\langle q \rangle ,a, \langle p\rangle } = \underline{\lambda}_{\langle q \rangle ,a,\langle p\rangle }$ for all $p \in W$ and $\underline{\mu}_{\langle q \rangle ,a,\langle p\rangle} = 0_S$ for all $\langle p \rangle \in Q' \setminus Q$;
		\item if $q \in W' \setminus W$, $a$ in $\Sigma$ and $qa$ is not in $W'$, then for each $\langle p \rangle \in Q'$ we define $\underline{\mu}_{\langle q \rangle ,a,\langle p \rangle}$ by induction on the length of $q$, as follows:\\
		If $q = \varepsilon$, then since $\underline{\lambda} \in \Lambda_Q$ for all words $w$ we have $\langle \varepsilon \rangle_w = \bigoplus_{\langle r \rangle \in Q} \underline{\lambda}_{\langle r \rangle} \otimes \langle r \rangle_w$ and so \[\langle a \rangle_w = \langle \varepsilon \rangle_{aw} =  \bigoplus_{\langle r \rangle \in Q} \underline{\lambda}_{\langle r \rangle} \otimes \langle r \rangle_{aw} =  \bigoplus_{\langle r \rangle \in Q} \underline{\lambda}_{\langle r \rangle} \otimes \langle ra \rangle_w =  \bigoplus_{\langle r \rangle, \langle p \rangle  \in Q} \underline{\lambda}_{\langle r \rangle} \otimes \underline{\lambda}_{\langle r \rangle ,a, \langle p \rangle} \otimes \langle p \rangle_w,\]
using the fact that $\langle r a \rangle = \langle r \rangle \cdot a$. In this case we set: $\underline{\mu}_{\langle \varepsilon \rangle,a,\langle p \rangle} = \bigoplus_{\langle r \rangle \in Q} \underline{\lambda}_{\langle r \rangle} \otimes \underline{\lambda}_{\langle{r} \rangle,a,\langle p\rangle}$ for all $\langle p \rangle \in Q$ and $\mu_{\langle \varepsilon \rangle ,a, \langle p \rangle } = 0_S$ for all $\langle p \rangle \in Q' \setminus Q$.\\
        If $q = ra$ and $\underline{\mu}_{\langle r \rangle ,a,\langle p \rangle}$ has already been defined for all $\langle p \rangle \in Q'$ in such a way that expresses $\langle q \rangle = \langle ra \rangle = \langle r \rangle \cdot a$ as a left-linear combination of the rows $Q$ (i.e. such that $\underline{\mu}_{\langle r \rangle,a,\langle p \rangle}=0_S$ if $\langle p \rangle$ is not in $Q$), then for all words $w$ and all $ b\in \Sigma$ we have 
		\begin{eqnarray*}(\langle q \rangle \cdot b)_w = \langle qb \rangle_w &=& \langle rab \rangle_w = \langle ra \rangle_{bw} =  \bigoplus_{\langle p  \rangle \in Q} \underline{\mu}_{\langle r \rangle ,a,\langle p\rangle}\otimes \langle p \rangle_{bw}\\
			&=& \bigoplus_{\langle p  \rangle \in Q} \underline{\mu}_{\langle r \rangle ,a,\langle p\rangle} \otimes \langle pb \rangle_{w} = \bigoplus_{\langle p \rangle,\langle s \rangle \in Q} \underline{\mu}_{\langle r \rangle ,a,\langle p\rangle} \otimes \underline{\lambda}_{\langle p \rangle ,b, \langle s\rangle} \otimes \langle s \rangle_{w}.\end{eqnarray*}
		In this case we set $\underline{\mu}_{\langle ra \rangle ,b,\langle s\rangle} = \bigoplus_{\langle p \rangle \in Q} \underline{\mu}_{\langle r \rangle,a,\langle p\rangle} \otimes \underline{\lambda}_{\langle p\rangle ,b,\langle s\rangle }$ for all $\langle s \rangle \in Q$ and $\mu_{\langle ra \rangle ,b,\langle s\rangle} = 0_S$ for all $\langle s \rangle \in Q' \setminus Q$.
	\end{itemize}
By construction we then have that $\underline{\mu}$ belongs to $\Lambda_{Q'}$ and hence by Lemma~\ref{theorem:prelim}, $\mathcal{H}_{Q',\underline{\mu}}$ computes $f$. The automaton $\mathcal{H}_{Q',\underline{\mu}}$ is also literal by construction, which concludes the proof.
	
(2. implies 3.) Let $\mathcal{A}$ be a literal automaton computing $f$ and $Q$ its set of states. By Theorem~\ref{theorem:prelim}, the set $\underline{Q}=\{\underline{q}  \mid q \in Q\}$, where $\underline{q}_w$ is defined as the value computed by the automaton $\mathcal{A}_q$ on input $w$, is row-generating for $f$ and, as observed in Remark \ref{rem:litrows}, there is a finite set of words $\{\sigma(q) \mid q \in Q\}$ such that $\underline{Q} = \{\langle\sigma(q) \rangle \mid q \in Q\}$. 
	
(3. implies 4.) Let $W \subseteq \Sigma^*$ be a finite set such that $\{\langle w \rangle \mid w \in W\}$ is row-generating and consider a sequence of left-semimodules $(X_i)_{i \in \mathbb{N}}$ as in the definition of the weak ascending chain condition. Then there exists $n$ such that for all $m \geq n$, all the rows $\langle w \rangle$ for $w$ in $W$ are in $X_m$, and since these elements are row-generating for $f$ it follows that $X_m = X_n$.

(4. implies 1.) For all positive integers $i$, let $X_i$ be the left-semimodule generated by the set $Q_i$ containing all rows of $F$ indexed by words of length at most $i$. Then the $X_i$ satisfy the hypothesis of the weak ascending chain condition and hence there is $n$ such that for all $m$, $X_m = X_n$. Necessarily the set $Q_n$ is row-generating and hence (by Lemma \ref{lem:leftgen}) left-closed, and left-generates $\langle \varepsilon \rangle$, so there exists some $\lambda \in \Lambda_Q$, giving $\Lambda_{Q_n} \neq \emptyset$ which demonstrates that $f$ is weakly guessable and concludes the proof.
\end{proof}

\begin{rem}
\label{rem:dual2}
As noted in Remark \ref{rem:dual}, one can also work with the transpose of the Hankel matrix, leading to obvious left-right dual notions which we refer to as weakly co-guessable, column generating sets, and weak co-ascending chain condition: see Appendix \ref{section:swap} for full details. We say that an automaton is co-literal if its mirror is literal. With these definitions in place, one can prove the following entirely analogous result.
\end{rem}

\begin{thm}
	\label{prop:coweakly}
	Given a function $f:\Sigma^* \to S$, the following assertions are equivalent:
	\begin{enumerate}
		\item $f$ is weakly co-guessable;
		\item there exists a co-literal automaton computing $f$;
		\item there exists a finite set of columns of the Hankel matrix of $f$ that is column-generating for $f$;
		\item $f$ satisfies the weak co-ascending chain condition.
	\end{enumerate}  
\end{thm}

\begin{rem}
\label{rem:fields2}
The semirings where \emph{all} functions computed by finite state weighted automata satisfy the weak ascending and co-ascending chain conditions are therefore the semirings for which all functions computed by finite state weighted automata are both weakly guessable and weakly co-guessable (by Theorem~\ref{prop:weakly} and its dual Theorem~\ref{prop:coweakly}). This means that over such a semiring all functions computed by a finite weighted automaton  have the potential to be learnt, provided one can make a lucky guess for $Q$, $T$ and $\lambda$. This is in particular the case for $\mathbb{B}$ (or indeed any finite semiring), $\mathbb{R}$ (or indeed, any field), and $\mathbb{Z}$  (or indeed, any principal ideal domain c.f. ~\cite{HKRS20}). More generally, if $R$ is any left-Noetherian ring, then all finitely generated left $R$-modules are Noetherian, and hence by Remark~\ref{rem:Noetherian}, \emph{all} functions computed by finite state weighted automata over $R$ are weakly guessable. This covers the case of finite rings, fields, principal ideal domains, and division rings.
\end{rem}

\subsection{Guessable and strongly guessable functions}

The definition of weakly guessable functions ensures the existence of some finite sets $Q,T \subseteq \Sigma^*$ and coefficients $\underline{\lambda} \in \Lambda_{Q,T}$ (writing here and from now on simply $Q$ in place of the set of rows of the Hankel matrix indexed by the words in $Q$, as discussed in Notation \ref{notation:lambda}) that together determine a correct hypothesis automaton. In other words, functions that are \emph{not weakly guessable} are those for which \emph{all hypothesis automata are incorrect}, and so (in fairness to the learner!) one should immediately exclude all such functions from consideration if one seeks to learn by membership and equivalence in the style of Angluin. For weakly guessable functions, the problem then becomes: how to find a suitable $Q, T$ and $\lambda$. Clearly $Q$ must be a row-generating set for the full Hankel matrix; but is it always possible to find such a set from finitely many membership and equivalence queries? Also, even if a suitable row generating set $Q$ is found, there is no guarantee that accompanying $T$ and $\lambda$ can be found in finite time. Guessable functions are those for which there exist finite subsets $Q, T \subseteq \Sigma^*$ such that all the associated hypothesis automata compute the correct function. Strongly guessable functions are those for which for \emph{every} finite row-generating set $Q$ there exists a finite subset $T \subseteq \Sigma^*$ such that all the associated hypothesis automata compute the correct function. In the following definition, the reader should recall from Remark \ref{rem:lambda} that the sets $\Lambda_{Q, T}$ depend upon the function $f$ (and hence also the semiring $S$).

\begin{defi}
A function $f: \Sigma^* \rightarrow S$ is guessable if there exist finite subsets $Q, T \subseteq \Sigma^*$ such that $\Lambda_{Q,T} = \Lambda_{Q} \neq \emptyset$. We say that the pair $(Q,T)$ witnesses that $f$ is guessable. Furthermore, we say that a guessable function $f$ is strongly guessable if for all finite sets $Q \subseteq \Sigma^*$, there exists a finite set $T \subseteq \Sigma^*$ such that $\Lambda_{Q,T} = \Lambda_{Q}$.
\end{defi}

It is clear from the definitions that every strongly guessable function is guessable, and every guessable function is weakly guessable. Moreover, observe that, by Theorem \ref{theorem:prelim}, if the pair $(Q, T)$ witnesses that $f$ is guessable, then all hypothesis automata over $Q$ and $T$ compute $f$.
We now give a weaker variant of another property given in~\cite{HKRS20} to characterise guessable functions.

\begin{defi}
Let $Q$ be a finite subset of $\Sigma^*$.
A total $Q$-chain is an infinite descending chain of left-solution sets $\Lambda_{Q, T_0} \supsetneq \Lambda_{Q, T_1} \supsetneq \Lambda_{Q, T_2} \supsetneq \cdots$  with $T_i \subseteq T_{i+1}$ for all $i \geq 0$ and such that each word in  $\Sigma^*$ is contained in some $T_i$.
We say that a function $f: \Sigma^* \rightarrow S$ is row-bound if there exists a finite set of words $Q$ such that 
$\Lambda_{Q} \neq \emptyset$
 and there is no total $Q-$chain.
We say that $f: \Sigma^* \rightarrow S$ is strongly row-bound if for \emph{every} finite set of words $Q$, there is no total $Q$-chain.
\end{defi}

The above definition is related to the progress measure defined in~\cite[Definition 11]{HKRS20}; we will discuss this in the next subsection.

\begin{prop}
\label{prop:guessable}
Let $f:\Sigma^* \to S$. Then $f$ is guessable if and only if $f$ is row-bound. Moreover, $f$ is strongly guessable if and only if $f$ is strongly row-bound and weakly guessable.
\end{prop}

\begin{proof} Suppose that $Q, T$ are finite sets of words such that $\Lambda_{Q, T}=\Lambda_{Q}$. Let $(T_i)_{i \geq0}$ be a sequence of subsets of $\Sigma^*$ with $T_i \subseteq T_{i+1}$ for all $i$ and such that $\bigcup_i T_i = \Sigma^*$. By definition, there exists some set $T_j$ such that $T \subseteq T_j$. Then for all $i\geq j$ we have $T \subseteq T_j \subseteq T_i$. Since $\Lambda_{Q, T}=\Lambda_{Q}$ it then follows that $\Lambda_{Q, T_i} = \Lambda_{Q, T_j} = \Lambda_{Q, T}$ for all $i \geq j$. Thus there can be no total $Q$-chains. It now follows that if $f$ is guessable then $f$ is row-bound, whilst if $f$ is strongly guessable, then it must be strongly row-bound (and weakly guessable).
	
	Now, let $Q$ be a finite set of words such that there is no total $Q$-chain. Then taking $T_i$ to be the finite set of words of length at most $i$  for each positive integer $i$, we have that $T_i \subseteq T_{i+1}$ and $\bigcup_i T_i = \Sigma^*$ and so there must exist $n$ such that for all $m\geq n$, $\Lambda_{Q, T_m} = \Lambda_{Q, T_n}$. Since the $T_i$ enumerate all words, we must therefore also have $\Lambda_{Q, T_n} = \Lambda_{Q}$.  Thus if $f$ is row-bound, there is such a $Q$ such that $\Lambda_Q \neq \emptyset$ and we conclude that $f$ is guessable. Additionally, if $f$ is strongly row-bound and weakly guessable then the property holds for all $Q$ and at least one of them is such that $\Lambda_Q \neq \emptyset$, which makes $f$ strongly guessable.
\end{proof} 

\begin{rem}
\label{rem:dual3}
As in Remarks \ref{rem:dual} and \ref{rem:dual2}, one can also work with the transpose of the Hankel matrix, leading to obvious left-right dual notions of co-guessable, strongly co-guessable: see Appendix \ref{section:swap} for full details. With these definitions, one can prove the following entirely analogous result.
\end{rem}

\begin{prop}
	\label{prop:coguessable}
	Let $f:\Sigma^* \to S$. Then $f$ is (strongly) co-guessable if and only if $f$ is (strongly) column-bound and weakly co-guessable.  
\end{prop}

\subsection{Algorithmic strategies}

\paragraph*{\textbf{Game strategies.}}

We now explain what the notions of (weakly/strongly) guessable mean in terms of strategies for the learner to discover an automaton correctly computing the target function via a sequence of Angluin-style queries.  

The hierarchy we propose is (in general) orthogonal to any complexity classes - in particular there is no guarantee that a strongly guessable function can be learnt in polynomial time. See Remark~\ref{rem:complexity} at the end of the section.

\begin{rem}
Provided the semiring S is countable, a trivial algorithm would be to enumerate all the automata and make an equivalence query on each of them in turn. One would of course like to improve upon this exponential time strategy. Even though there is no guarantee in terms of time complexity of the Angluin-style process we describe next, note that there are two main differences: (1) the learning process still applies when the semiring is uncountable, (2) it will only explore, in a prescribed manner, a subset of automata i.e. the class of hypothesis automata. 
\end{rem}

Consider a two player game between the teacher and the learner, in which the learner's objective is to find an automaton computing a target function. The rules are as follows. At every round, the learner is allowed to:
\begin{itemize}
\item choose a pair of finite subsets $Q, T \subseteq \Sigma^*$; 
\item ask the teacher to provide an hypothesis automaton constructed from $Q$ and $T$.
\end{itemize}
 The teacher must:
\begin{itemize} 
    \item declare if there are no possible hypothesis automaton for a given pair $Q,T$ (that is, if $\Lambda_{Q,T} = \emptyset$) or provide a hypothesis automaton $\mathcal{H}_{Q, T, \underline{\lambda}}$ for this pair (by selecting $\underline{\lambda} \in \Lambda_{Q, T}$); 
    \item declare whether each provided automaton computes the target function and if not, provide a counter-example.
\end{itemize}\noindent 
The game stops when the teacher declares an automaton that computes the target function.  We view the teacher as an oracle, able to determine an element of each of the sets $\Lambda_{Q, T}$, $\Lambda_{Q}$ and $\Lambda_{Q, T}\setminus\Lambda_{Q}$ (for all finite subsets of words $Q, T$), or else determine that the set is empty. In this way, all computation is deferred to the teacher. The teacher must always answer truthfully, but may play as an ally (meaning whenever $(Q,T)$ is queried with $\Lambda_{Q}\neq \emptyset$, the teacher will choose to reply with a correct automaton $\mathcal{H}_{Q, T, \underline{\lambda}}$ where $\underline{\lambda} \in \Lambda_{Q}$) or as an adversary (whenever $(Q,T)$ is queried with $\Lambda_{Q,T} \neq \Lambda_{Q}$, the teacher will choose to reply with an incorrect automaton $\mathcal{H}_{Q, T, \underline{\lambda}}$ where $\underline{\lambda} \in \Lambda_{Q,T} \setminus \Lambda_{Q}$).

One can record the history of a game up to any given point  as the finite (possibly empty) sequence of pairs $(Q,T)$ queried so far, together with the replies from the teacher (i.e. the hypothesis automata and counter-examples or the fact that $\Lambda_{Q,T} = \emptyset$). A strategy for the learner is a function mapping each history $h$ where the last pair in the sequence does not signify the end of the game to a new pair $(Q_h,T_h) \in \Sigma^* \times \Sigma^*$. Additionally, the strategy is said to be incremental if for all histories $h$, $Q_h \supseteq Q$ and $T_h \supseteq T$ holds for all pairs $(Q,T)$ already queried in $h$.     
A strategy is winning for the learner, if, for all possible plays following that strategy, the teacher eventually declares that a correct automaton is found. 

\begin{prop}
A function is weakly guessable if and only if the learner has a winning strategy in the game where the teacher plays as an ally. 
\end{prop}

\begin{proof}
For a weakly guessable function there exists a non-empty finite subset $Q \subseteq \Sigma^*$ such that $\Lambda_{Q} \neq \emptyset$. Since $\Sigma^*$ is countable, so is the set of non-empty finite subsets of $\Sigma^*$. Let $(Q_i)_{i\in \mathbb{N}}$ be an enumeration of the finite non-empty subsets of $\Sigma^*$ and consider the following strategy for the learner: if $h$ is empty, query $(Q_0, \{\varepsilon\})$; if the last query of $h$ is $(Q_i, \{\varepsilon\})$ and does not halt the game then query $(Q_{i+1}, \{\varepsilon\})$ (no matter the automata and counter-examples given by the teacher). The learner will eventually query $(Q,\{\varepsilon\})$ and since $\Lambda_{Q} \subseteq \Lambda_{Q,\{\varepsilon\}}$ the teacher (playing as an ally) will choose to declare a correct hypothesis automaton. Conversely, if the learner has a winning strategy, then clearly there exist finite sets $Q, T \subseteq \Sigma^*$ and $\underline{\lambda} \in \Lambda_{Q, T}$ such that $\mathcal{H}_{Q, T, \underline{\lambda}}$ computes the target function and so by definition the function is weakly guessable.
\end{proof}

\begin{prop}
\label{prop:guessgame}A function is guessable if and only if the learner has a winning strategy in the game where the teacher plays as an adversary. 
\end{prop}

\begin{proof}
For a guessable function there exists finite subsets $Q, T \subseteq \Sigma^*$ such that $\Lambda_{Q,T} = \Lambda_{Q} \neq \emptyset$. Since $\Sigma^*$ is countable, so is the set of pairs of non-empty finite subsets of $\Sigma^*$; let $(Q_i, T_i)_{i\in \mathbb{N}}$ be an enumeration of all such pairs and consider the following strategy for the learner: if the history is empty, query  $(Q_0,T_0)$; if the last query of $h$ is $(Q_i, T_i)$ and does not halt the game then query $(Q_{i+1}, T_{i+1})$ (no matter the automata and counter-examples given by the teacher). The learner will eventually query $(Q,T)$ and since $\Lambda_{Q} = \Lambda_{Q,T}$, the teacher will be forced to declare a  correct hypothesis automaton.

Conversely, if the learner has a winning strategy in the game where the teacher plays as an adversary, then clearly there exist finite sets $Q, T \subseteq \Sigma^*$ and $\underline{\lambda} \in \Lambda_{Q, T}$ such that $\mathcal{H}_{Q, T, \underline{\lambda}}$ computes the target function and $\Lambda_{Q, T} \setminus \Lambda_Q = \emptyset$ (otherwise the teacher can always choose to provide an incorrect automaton). So by definition the function is guessable.
\end{proof}

\begin{prop}
\label{prop:game}If a function is strongly guessable, then the learner has a winning strategy that is incremental in the game where the teacher plays as an adversary.
\end{prop}

\begin{proof}
Let $Q_i = T_i$ be the set of all words of length at most $i$, and consider the following strategy for the learner: if the history is empty, query $(Q_0,T_0)$; if $(Q_i, T_j)$ is the last pair queried in $h$ and the teacher declared that $\Lambda_{Q_i,T_j}$ is non-empty, but did not give a correct automaton, then query $(Q_i, T_{j+1})$; if $(Q_i, T_j)$ is the last pair queried in $h$ and $\Lambda_{Q_i,T_j}$ is empty then query $(Q_{i+1}, T_{j})$; otherwise if $(Q_i, T_j)$ is the last pair queried in $h$ and the teacher declared a correct automaton, the game halts.

For any given target function this process is completely deterministic, but possibly infinite: since the teacher plays as an adversary (and so never supplies a correct automaton if there is an incorrect choice), the sequence of pairs queried by the learner for this function will be the same regardless of different choices of automata or counterexample supplied by the teacher in response to each query (only the question of their existence is relevant to proceedings). Our aim is to show that whenever the target function is strongly guessable, this process terminates. 

Since $f$ is weakly guessable and for all non-negative integers $i$, $Q_i \subseteq Q_{i+1}$ then there exists a non-negative integer $I$ such that for all $i< I$, $\Lambda_{Q_i}$ is empty and $\Lambda_{Q_I}$ is non-empty. Hence by strong guessability, for all $i< I$, there exists $T$ such that $\Lambda_{Q_i,T}$ is empty and there exists $T'$ such that $\Lambda_{Q_I,T'} = \Lambda_{Q_I}$ is non-empty. By Remark \ref{rem:subsets} and Lemma \ref{lem:Qincrement}, it follows that for all $i< I$, there exists $\ell$ such that for all $j\geq \ell$, $\Lambda_{Q_i,T_j}$ is empty and $J$ such that for all $j\geq J$, $\Lambda_{Q_I,T_j} = \Lambda_{Q_I}$ is non-empty (and none of the $\Lambda_{Q_I,T}$ is empty for any $T$).
Hence, the process terminates, and the strategy above is winning for the learner.
\end{proof}

\begin{rem}
We leave open the question whether the converse of Proposition~\ref{prop:game} holds. This is related to the question of the existence of functions that are guessable but not strongly guessable. This question remains open (see Section~\ref{subsec:strict} and Conclusion).
\end{rem}

\begin{rem}
\label{rem:complexity}
The strategies mentioned above do not guarantee any polynomial time complexity.
For strongly guessable functions, a polynomial time strategy would require the learner to make more well-informed choices of $Q$ and $T$ (by making use of the counter-examples) so that the size of $Q$ and $T$ remain polynomial with respect to the size of the automaton to learn and counter-examples given.  For guessable and weakly guessable functions, in addition to the constraints on the size of $Q$ and $T$, the learner needs to be able to find a correct $\underline{\lambda}$ in polynomial time; this may be possible over certain semirings.
\end{rem}

\paragraph*{\textbf{Relationship to the algorithmic approach of \cite{HKRS20}.}}

In~\cite{HKRS20}, a general learning algorithm is given (subject to the existence of a `closedness strategy') and three conditions (solvability, ascending chain condition, and progress measure) are requested to ensure termination and correctness of the algorithm \emph{for all functions computed by weighted automata over a fixed semiring}.  We give an overview of this algorithm using our terminology in Figure~\ref{fig:algo}, and discuss below some weaker conditions that are easily seen to be sufficient for the proof of termination and correctness \emph{on a given target function}, and how these relate to our definitions of guessability.

\begin{figure}[ht]
	\hrule
	\noindent
	\flushleft	\textbf{Learning Algorithm}
	\hrule\medskip\noindent
	\textbf{Initialisation:} Set $Q = T =\{\varepsilon\}$.\\
	\textbf{Closure:} While $Q$ is not left-closed on $T$, do: \\
	\begin{tabular}{c | l}
		\hspace{0.5cm} &
		\begin{minipage}{0.9\textwidth}
			Find $q \in Q$, $a \in \Sigma$ such that $\langle qa \rangle_T$ is not a left-linear combination of the rows in $Q$.\\ Set $Q:=Q \cup \{qa\}$.
		\end{minipage}
	\end{tabular}
	
	\textbf{Test:} Find $\lambda$ in $\Lambda_{Q,T}$ and perform an equivalence query for $\mathcal{H}_{Q,T,\lambda}$.\\
    If the oracle gives a counter-example $z$, set $T = T \cup \{\text{suffixes of } z\}$, and go back to the Closure step. Else terminate with success.
	\caption{Outline of learning algorithm from \cite{HKRS20}, subject to the existence of a closedness strategy.}
	\label{fig:algo}
\end{figure}

First, note that this algorithm requires  access to a procedure (referred to as a `closedness strategy' in \cite{HKRS20}) that can (a) determine whether $Q$ is left-closed on $T$ (or in other words, whether $\Lambda_{Q,T} \neq \emptyset$); (b) if $\Lambda_{Q,T} = \emptyset$, can find $q \in Q$ and $a \in \Sigma$ such that the corresponding equation has no solution; and (c)  if $\Lambda_{Q,T} \neq \emptyset$, can find an element $\underline{\lambda}$ of $\Lambda_{Q, T}$. 

Clearly, if $S$ is a semiring where \emph{there exists} an algorithm which can take as input an arbitrary system of linear equations over $S$ and compute a solution, then there is a closedness strategy for all functions over this semiring.  We note that in \cite{HKRS20} such semirings were called \emph{solvable}. It is clear that all finite semirings are solvable (a very naive algorithm would be to try all possibilities) and fields where the basic operations are computable are also solvable (via Gaussian elimination). In \cite{HKRS20} the authors provide a sufficient condition for a principal ideal domain to be solvable (specifically: all ring operations are computable, and each element can be effectively factorised into irreducibles). 

\begin{rem}Turning to other examples,  we note over $\mathbb{Z}_{\rm max}$ there is a straightforward (and well-known)  method to determine whether a system of the form $A \otimes \underline{x} = \underline{b}$ has a solution. Indeed, if $\underline{x}$ is a solution, we must in particular have that $A_{i,j} + \underline{x}_j \leq \underline{b}_i$ for all $i,j$. If $A_{i,j} = -\infty$,  then this inequality holds automatically, otherwise, we require that $\underline{x}_j \leq \underline{b}_i - A_{i,j}$. In order to solve the inequality, for each $j$ such that there exists at least one $i$ with $A_{i,j}\neq -\infty$ we could set $\underline{x}_j = {\rm min}\{\underline{b}_i-A_{i,j}: A_{i,j} \neq -\infty\}$, and set all remaining $\underline{x}_j$ equal to $0$. It is clear that this is a computable (we use finitely many operations of subtraction and comparison on finitely many entries drawn from $\mathbb{Z}_{\rm max}$) solution to the system of inequalities $A \otimes \underline{x} \leq \underline{b}$. Furthermore, it can also be seen that either this solution  is also a solution to the equation $A \otimes \underline{x} = \underline{b}$, or else no solution to the equation $A \otimes \underline{x} = \underline{b}$ exists. (Indeed, if the solution to the inequality that we have constructed does \emph{not} satisfy the equation, then we must have that for some $i$ the following strict inequality holds: ${\rm max} (A_{i,j} +\underline{x}_j) < \underline{b}_i$, meaning that either all $A_{i,j}$ are equal to $-\infty$ and so there is no hope to satisfy the corresponding equation, or else in order to attain the equality we require that one of the $\underline{x}_j$ for $j$ with $A_{i,j} \neq -\infty$ must take a larger value than set above, which by construction is not possible.) Thus $\mathbb{Z}_{\rm max}$ (and likewise $\mathbb{N}_{\rm max}$) is solvable.
\end{rem}

\begin{rem} 
In~\cite{HKRS20}, an example showing that the algorithm in Figure~\ref{fig:algo} does not always terminate is given. This example corresponds to the automaton given in Figure~\ref{fig:automata4}, that (as we shall see in the next section) is not weakly guessable. Since the algorithm can only guess hypothesis automata, it is clear that the algorithm will not terminate on any function that is not weakly guessable. In fact, termination  cannot be ensured if the target function is not guessable. Indeed, if the function is not guessable, then (by definition) for each finite set of words $Q$ with $\Lambda_Q \neq \emptyset$ one has that $\Lambda_Q \neq \Lambda_{Q, T}$ for all finite sets of words $T$. This in turn means that there is no guarantee that the algorithm will terminate, since it could be that in \emph{every test step} the  element $\underline{\lambda}$ provided by the closedness strategy generates an incorrect hypothesis automaton. \end{rem}

For the remainder of this section, let us assume that $f$ is a function for which a closure strategy exists. The following  condition ensures that the closure step terminates in finite time for each fixed $T$ and that it cannot be taken infinitely many times for different $T$'s.

\begin{defi}
	A function $f:\Sigma^* \to S$ with Hankel matrix $F$ satisfies the ascending chain condition on rows if for all sequences of left-semimodules $(X_i)_{i \in \mathbb{N}}$ generated by finite sets of rows of $F$, such that $X_0 \subseteq X_1 \subseteq X_2 \subseteq \cdots$
	there is $n$ such that for all $m \geq n$, $X_m = X_n$.
\end{defi}

Finally, we note that the following condition (which mimics the definition of progress measure given in \cite{HKRS20}, rewritten in the language of the present paper and extracting the key property) ensures that one cannot go through the Test step infinitely many times without changing the set $Q$ in the closure step. 

\begin{defi}
	Let $f$ be a function computed by a weighted automaton over $S$ and let
    $$\mathcal{L}: =\{(Q,T): Q, T \mbox{ are finite sets of words and } Q \mbox{ is left-closed on }T\}.$$
     We say that $f$ has a weak progress measure if there is a totally ordered set $(\mathbb{K},\prec)$ with no infinite descending chain and a function $\mu: \mathcal{L} \rightarrow \mathbb{K}$ with the property that if $T \subseteq T'$ and $\Lambda_{Q, T} \supsetneq \Lambda_{Q, T'}$  we have $\mu(Q, T)\succ  \mu(Q, T')$. 
\end{defi}

We summarise the observations above in the following proposition, which provides a slight generalisation of statement of the main result of \cite{HKRS20}; the proof derives in a similar manner to results presented in~\cite{HKRS20}.

\begin{prop}\cite[Theorem 14]{HKRS20}
Suppose that a closure strategy exists for $f$. If $f$ satisfies the ascending chain condition on rows and has a weak progress measure, then the learning algorithm from Figure~\ref{fig:algo} terminates and returns an automaton computing $f$.
\end{prop}

\begin{rem}
\label{rem:terminate}
Finally, we note the relationship between our notions of guessability and the termination conditions stated above.
\begin{itemize}
\item The ascending chain condition on rows is weaker than the ascending chain condition considered in \cite{HKRS20} (which considers \emph{all} chains of submodules, including those which may not be generated by finitely many \emph{rows} of the matrix) but stronger than the weak ascending chain condition from Definition~\ref{definition:wacc} (which considers only \emph{certain} chains of submodules generated by finitely many rows of the matrix). Thus a function satisfying the termination conditions above is certainly (by Theorem \ref{prop:weakly})  weakly guessable. Weak guessability alone is not sufficient to ensure termination of algorithm from Figure~\ref{fig:algo} (as we have seen, it is precisely the condition required in order for a correct hypothesis automaton to exist).
\item  The existence of a progress measure (in the sense of \cite{HKRS20}) for a Hankel matrix implies that the corresponding function is strongly row-bound. However, the converse need not be true. A progress measure implies a uniform bound, for all $Q$, on the number of elements in a total $Q$-chain, but for a (strongly) row-bound function, even if no such chain can be infinite, its size might depend on the set of words $Q$. 
\item Each progress measure gives rise to a weak progress measure, and moreover, if a function has a weak progress measure then it is strongly row-bound or equivalently (by Proposition \ref{prop:guessable}) strongly guessable. Thus the proposition above guarantees that the learning algorithm in Figure \ref{fig:algo} terminates for a class of strongly guessable functions. Strong guessability alone is not sufficient to ensure termination of algorithm from Figure~\ref{fig:algo}. However, it follows from the proof of Proposition  \ref{prop:game} that there is an algorithm with a different (incremental) update strategy for the sets $Q$ and $T$ that will terminate with success on any strongly guessable function for which a closure strategy exists.
\end{itemize}
\end{rem}

\section{Hierarchy of guessable functions}
\label{section:hierarchy}

In the previous section, we have defined the classes of strongly guessable, guessable, and weakly guessable functions, and in Appendix \ref{section:swap} the dual notions of strongly co-guessable, co-guessable, and weakly co-guessable functions. 

\subsection{The general hierarchy}

First, we show that for any semiring: if a function is both weakly and co-weakly guessable then it is also both strongly and co-strongly guessable, so that the general hierarchy is as depicted in Figure~\ref{fig:hierarchy2}.

\begin{figure}[ht]
\includegraphics[width=\textwidth]{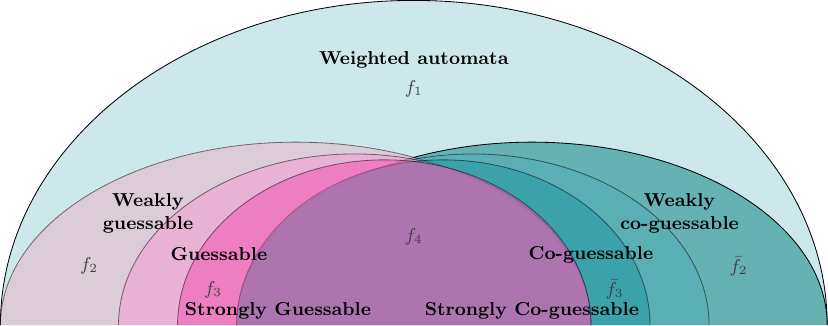}
\caption{A general hierarchy for levels of guessability}
\label{fig:hierarchy2}
\end{figure}

\begin{lem}
\label{lemma:com}
Let $S$ be a semiring and let $f: \Sigma^* \rightarrow S$ be computed by a weighted automaton over $S$. If $f$ is weakly guessable and weakly co-guessable, then $f$ is strongly guessable and strongly co-guessable.
\end{lem}

\begin{proof}
	Suppose that $f$ is weakly guessable and weakly co-guessable and let $F$ be the Hankel matrix of $f$. We first show that $f$ is strongly guessable. Since $f$ is weakly co-guessable, we may fix a  finite set of words $T$ such that the columns indexed by $T$ are column-generating. We will prove that for any finite set of words $Q$, we have $\Lambda_{Q} = \Lambda_{Q,T}$.
	
	Thus let $Q$ be arbitrary. Since $\Lambda_{Q} \subseteq \Lambda_{Q,T}$ is always true, if $\Lambda_{Q,T} = \emptyset$ then the equality holds. Otherwise, let $\underline{\lambda} \in \Lambda_{Q, T}$. Then for all $q, p \in Q$, all $a \in \Sigma$ and all $t \in T$ we have:
	$F_{qa,t} = \bigoplus_{p \in Q} \underline{\lambda}_{q,a,p} \otimes F_{p,t}$. By our assumption on $T$ (recalling that we consider \emph{right} linear combinations in the definition of weakly co-guessable) we know that for all words $w$ and all $t \in T$ there exist $\mu_{w,t} \in S$ such that $F_{u, w} = \bigoplus_{t \in T} F_{u, t} \otimes \mu_{w,t}$ holds for all words $u$. In particular, for all $w \in \Sigma^*$, we have:
	\begin{eqnarray*}
		F_{qa, w} &=& \bigoplus_{t \in T} F_{qa, t} \otimes \mu_{w,t} = \bigoplus_{t \in T} (\bigoplus_{p \in Q} \underline{\lambda}_{q,a,p} \otimes F_{p, t}) \otimes \mu_{w,t} \\
		&=& \bigoplus_{p \in Q} \underline{\lambda}_{q,a,p} \otimes (\bigoplus_{t \in T}  F_{p, t}\otimes \mu_{w,t})\\
		&=& \bigoplus_{p \in Q}\underline{\lambda}_{q,a,p} \otimes F_{p, w}.
	\end{eqnarray*}
	In other words, $\langle qa \rangle = \bigoplus_{p \in Q}\underline{\lambda}_{q,a,p} \otimes \langle p \rangle$. A similar argument shows that $\langle \varepsilon \rangle = \bigoplus_{q \in Q}\underline{\lambda}_{q} \otimes \langle q \rangle$, hence giving $\underline{\lambda} \in \Lambda_Q$, as required. This demonstrates that $f$ is strongly guessable. A dual argument (interchanging the rows/columns, and left/right combinations) demonstrates that $f$ is strongly co-guessable.
\end{proof}

\begin{rem}
\label{rk:universalsg}
The proof of Lemma~\ref{lemma:com} actually shows that any function that is both weakly guessable and weakly co-guessable has the following stronger property. Say that a guessable function $f: \Sigma^* \rightarrow S$ is uniformly strongly guessable if there exists a finite set $T \subseteq \Sigma^*$ such that for all finite sets $Q \subseteq \Sigma^*$ one has $\Lambda_{Q,T} = \Lambda_{Q}$. Note the inversion of quantifiers as compared to the definition of strongly guessable; here the set $T$ has to be the same for all the sets $Q$. Dually one can define uniformly strongly co-guessable. The previous proof shows that the intersection of the class of weakly guessable and weakly co-guessable functions coincides with the intersection of the class of uniformly strongly guessable and uniformly strongly co-guessable functions.  We leave open the question of whether this intersection  is precisely characterised by uniform strong guessability. However, we note that it can be shown that the class of uniformly strongly guessable functions is strictly contained in the class of strongly guessable functions (as an example, the function $f_3$ discussed below is not uniformly strongly guessable).
\end{rem}

\subsection{Examples for the strict hierarchy.}
\label{subsec:strict}

We now give examples to demonstrate that certain inclusions in the hierarchy can be strict. The examples are drawn from various well-studied semirings discussed in Section~\ref{section:prelim}, but one could obtain a unified semiring separating these classes by considering their product with component-wise operations.  The full proofs of the membership of the following functions to certain parts of the hierarchy are given in Appendix~\ref{appendix:examples}. We give here a quick overview making frequent use of the following lemma:

\begin{lem}
\label{lem:com2}
If $S$ is a commutative semiring then a function is weakly guessable (respectively, guessable, strongly guessable) if and only if its mirror is weakly co-guessable (respectively, co-guessable, strongly co-guessable).
\end{lem}
\begin{proof}
Let $f: \Sigma^* \rightarrow S$, and consider its mirror $\overline{f}: \Sigma^* \rightarrow S$. Let us write $F$ to denote the Hankel matrix of $f$ and $\overline{F}$ to denote the Hankel matrix of $\overline{f}$, and for a set of words $W \subseteq \Sigma^*$ let us also write $\overline{W} = \{\overline{w}: w \in W\}$. By definition, for all words $u, v \in \Sigma^*$ we have $$\overline{F}_{u,v} = \overline{f}(uv) = f(\overline{uv}) = f(\overline{v} \; \overline{u}) = F_{\overline{v}, \overline{u}}.$$
From this one sees that for any pair of finite subsets $Q, T \subseteq \Sigma^*$, the equations given in Notation \ref{notation:lambda} for the function $f$ (encoding certain left-linear combinations of the rows of $F$ indexed by $Q$ restricted to columns in $T$)
are, due to commutativity of multiplication in $S$, equivalent to the equations  encoding (right) linear combinations of the columns of $\overline{F}$ indexed by $\overline{Q}$ restricted to rows in $\overline{T}$, as obtained from \ref{nota:gamma}. Since weakly guessable, guessable and strongly guessable and their duals are described in terms of properties of these solution sets (considered for all finite sets $Q, T$), the result follows. 
\end{proof}

For some semirings $S$, there are functions $f: \Sigma^* \rightarrow S$ that are neither weakly guessable, nor weakly co-guessable. We show that this is the case for $f_1$ computed by the automaton given in Figure~\ref{fig:automata1}, whether viewed over $\mathbb{N}_{\max}$, $\mathbb{Z}_{\max}$ or $\mathbb{R}_{\max}$, $f'_1$ computed by the automaton given in Figure~\ref{fig:automata4}, whether viewed over $\mathbb{N}$ or $\mathbb{R}_{\geq 0}$, and $f''_1$ computed by the automaton in Figure~\ref{fig:automata4bis}. There are also functions that are weakly guessable but neither guessable nor weakly co-guessable. We show that this is the case for $f_2$ computed by the automaton given in Figure~\ref{fig:automata3}, whether viewed as an automaton over $\mathbb{Z}_{\max}$ or $\mathbb{R}_{\max}$. It then follows that the mirror function $\bar{f_2}$ is weakly co-guessable but neither co-guessable nor weakly guessable (by Lemma \ref{lem:com2}). Finally, we show that the functions $f_3$ computed by the automaton depicted in Figure~\ref{fig:automata3} (over $\mathbb{N}_{\max}$), $f'_3$ computed by the automaton from Figure~\ref{fig:automata5} and $f''_3$ computed by the automaton from Figure~\ref{fig:automata5bis} are all strongly guessable but not weakly co-guessable. The functions $\bar{f_3}$ and $\bar{f'_3}$ are strongly co-guessable but not weakly guessable. 
It is in general easy to find examples at the intersection of all the classes: for commutative semirings, any function computed by an automaton that is both literal and co-literal lies in this intersection. A trivial example of such function would be the constant function $f_4$ mapping all words to $1_S$.

\subsection{Collapse in some familiar semirings}
In many well-studied semirings, we note that the hierarchy collapses further, due to the fact that certain ascending and descending chain conditions are satisfied. It is clear from Theorem \ref{prop:weakly}, Theorem \ref{prop:coweakly}, Proposition \ref{prop:guessable}, Proposition \ref{prop:coguessable} and Lemma \ref{lem:com2} that the following two properties characterise (A) when all functions computed by a finite state weighted automaton over $S$ are strongly guessable and strongly co-guessable (i.e. everything collapses, as in Figure \ref{fig:hierarchy3}) and (B) when all weakly guessable functions are strongly guessable and all weakly co-guessable functions are strongly co-guessable (as in Figure \ref{fig:hierarchy4}):
\begin{itemize}
	\item Property A: all functions computed by finite state weighted automata over $S$ satisfy the weak ascending and weak co-ascending chain condition;
	\item Property B: all functions computed by finite state weighted automata over $S$ which satisfy the weak ascending chain condition are strongly row-bound and all functions computed by finite state weighted automata over $S$ which satisfy the weak co-ascending chain condition are strongly column-bound.
\end{itemize}

\begin{figure}[ht]
     \centering
     \begin{subfigure}[b]{0.4\linewidth}
         \centering
         \includegraphics[width=\textwidth]{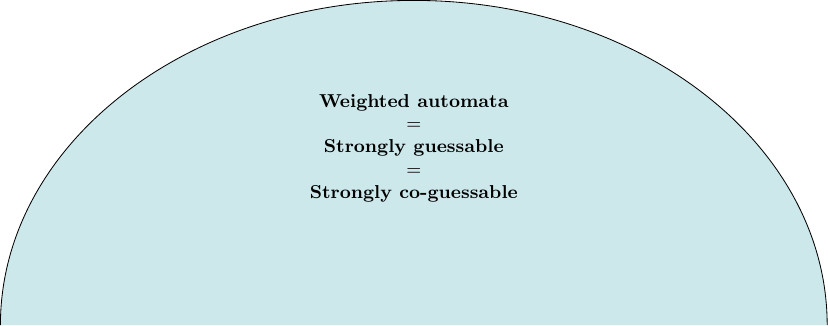}
         \caption{Property A}
         \label{fig:hierarchy3}
     \end{subfigure}
     \hspace{1cm}
     \begin{subfigure}[b]{0.4\linewidth}
         \centering
         \includegraphics[width=\textwidth]{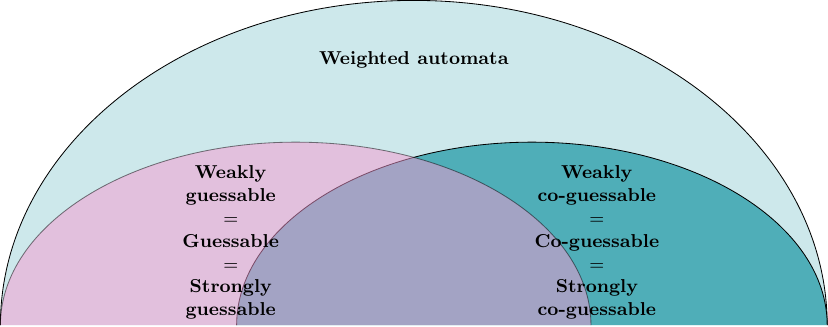}
         \caption{Property B}
         \label{fig:hierarchy4}
     \end{subfigure}
\caption{Hierarchies for semirings satisfying properties A and B.}
\label{fig:hierarchies}
\end{figure}
\noindent 
Note that Property A holds in the semirings $\mathbb{Z}$, $\mathbb{R}$ and $\mathbb{B}$ (as special cases of principal ideal domains, fields and finite semirings respectively) for example. Further, the semirings $\mathbb{N}_{\max}$,  $\mathbb{N}$ and $\mathcal{P}_{\rm fin}(\Sigma^*)$ satisfy Property B:	 

	\begin{lem}
		\label{lemma:PM}
		Every weakly guessable (respectively, weakly co-guessable) function over $\mathbb{N}_{\max}$, $\mathbb{N}$, or $\mathcal{P}_{\rm fin}(\Sigma^*)$ is strongly row-bound (respectively, column-bound).
	\end{lem}

	\begin{proof}
	We give the proof for rows; the proof for columns is dual. Let $S$ be one of the semirings $\mathbb{N}_{\max}$, $\mathbb{N}$, or $\mathcal{P}_{\rm fin}(\Sigma^*)$.
	Let $f$ be a weakly guessable function computed over $S$, and $F$ its Hankel matrix. Let $Q \subseteq \Sigma^*$ be a finite set. We aim to show that there is no total $Q$-chain.
	First, note that, for some fixed $T$, if no row indexed by $Q$ and restricted to $T$ is fully $0_S$, then in each of three semirings considered, $\Lambda_{Q,T}$ is finite (possibly empty) from which it immediately follows that there is no total $Q$-chain involving $\Lambda_{Q,T}$. Let $Z_{T} = \{q \in Q \mid \langle q \rangle_T \mbox{ is a }0_S \mbox { row}\}$, $Q' = Q \setminus Z_{T}$ and $N_{T} \subseteq S^{Q' \cup (Q \times \Sigma \times Q')}$ be the set obtained from $\Lambda_{Q,T} \subseteq S^{Q \cup (Q \times \Sigma \times Q)}$ by restricting the entries $\underline{\lambda}_p$ and $\underline{\lambda}_{q,a,p}$ to only permit $p \in Q'$. Clearly $Z_{T}$ is finite and its size is bounded above by $|Q|$. Notice that $N_{T}$ is also a finite set, since each  $\underline{\lambda}_p, \underline{\lambda}_{q,a,p}$ where $p \in Q'$ can only take a finite number of different values. 
	Moreover, it is clear from the definitions that $T \subseteq T'$ forces $Z_{T}\supseteq Z_{T'}$ (i.e. any row of the Hankel matrix that is zero on all columns indexed by $T'$ is also zero on all columns indexed by $T$), $\Lambda_{Q, T} \supseteq \Lambda_{Q, T'}$ (by Remark \ref{rem:subsets}) and $N_{T} \supseteq N_{T'}$ (since these sets are obtained from the previous by restricting attention to coefficients of rows indexed by $Q$ that are not identically zero on $T$ and $T'$ respectively).  Moreover, notice that if $Z_{T}=Z_{T'}$ (the same subset of elements of $Q$ are zero when restricted to $T$ or $T'$) and $N_T=N_{T'}$ then we must have $\Lambda_{Q, T}=\Lambda_{Q, T'}$. Now, suppose that $T \subsetneq T'$ with  $\Lambda_{Q, T} \supsetneq \Lambda_{Q, T'}$. It is clear from the above that $Z_{Q, T}\supseteq Z_{Q, T'}$, and either this containment is strict, or if not, then $N_{Q', T} \supsetneq N_{Q', T'}$. Since both sets have a finite size, there cannot be an infinite descending chain as in the definition.
\end{proof}

 We remark that in semirings where the concepts of weakly guessable and strongly guessable (resp. weakly co-guessable are strongly co-guessable) coincide (as is the case for fields and finite semirings, for example), this is inherited by subsemirings:
	
		\begin{lem}
			\label{prop:subsemiring}
			Let $S'$ be a subsemiring of $S$. If every weakly guessable function over $S$ is strongly guessable, then every weakly guessable function over $S'$ is also strongly guessable.
   In particular, if $S$ is a subsemiring of a field (e.g. an integral domain), then every weakly guessable function over $S$ is strongly guessable.
		\end{lem}
	
		\begin{proof}
	Suppose that $f: \Sigma^* \rightarrow S'$ is weakly guessable. Thus there is a literal automaton $\mathcal{A}$ with weights in $S'$ that computes $f$. Viewed as a function over $S$, we also have that $f$ is weakly guessable (since $\mathcal{A}$ is a literal automaton with weights in $S$ computing $f$) and hence by assumption, strongly guessable over $S$. Suppose for contradiction that $f$ is not strongly guessable over $S'$. Then there exists a finite set $Q$ such that for all finite sets $T$ there exists $\underline{\lambda}_T \in \Lambda_{Q, T} \setminus \Lambda_Q$, where these solution sets are considered over $S'$. However, noting that the coefficients $\underline{\lambda}_T$ lie in $S$, this directly contradicts that $f$ is strongly guessable over $S$. The final statement follows from Remark~\ref{rem:fields2} together with the fact that any integral domain embeds into a field (its field of fractions).
\end{proof}	

\begin{rem}
Suppose $S'$ is a subsemiring of $S$ and $f: \Sigma^* \rightarrow S'$. Then we may also regard $f$ as a function with co-domain $S$, and as observed in the previous proof, if $f: \Sigma^* \rightarrow S'$ is weakly guessable then $f$ is also weakly guessable when regarded as a function with co-domain $S$. However, if $f$ is guessable (respectively, strongly guessable) when considered as function with co-domain $S'$, it need not be guessable (respectively, strongly guessable) when considered as a function with co-domain $S$. For example, the function computed by the automaton depicted in Figure \ref{fig:automata1} is strongly guessable when the weights (and hence the co-domain of the function) are considered to lie in the semiring $S'=\mathbb{N}_{\rm max}$, but it is \emph{not} guessable when the weights are considered to lie in $S=\mathbb{Z}_{\rm max}$ (see the functions denoted $f_2$ and $f_3$ in Appendix \ref{app:examplef3} for details). 
\end{rem}

We now apply the results and discussion of this section to determine whether Property A or B holds in the semirings used in our examples. The table in Figure \ref{fig:table} demonstrates that there are interesting and important semirings where Property B holds (i.e. all weakly (co-)guessable functions are strongly (co-)guessable) whilst Property A does not hold (i.e. there exist functions that are not weakly guessable). Since over such semirings the property of being \emph{strongly} guessable is  characterised in terms of automata (by Theorem~\ref{prop:weakly}), it is possible for the `teacher' to select a strongly guessable target function to begin with, and in light of the discussion in Remark \ref{rem:terminate}, there is then an algorithm (similar to that of \cite{HKRS20}, but using a different update strategy) that terminates with success for all such target functions with a closure strategy.

\begin{figure}[ht]
	\begin{minipage}{6.5cm}
		\begin{tabular}{c|cc}
			\multicolumn{1}{c}{} & A & B  \\ \toprule
			$\mathbb{B}$ & Yes&  Yes \\ 
			$\mathbb{R}$ & Yes &  Yes  \\ 
			$\mathbb{Z}$ & Yes&  Yes \\ 
			$\mathbb{N}$ & No ($f'_1$) &  Yes   \\ 
			$\mathbb{R}_{\geq 0}$ & No ($f'_1$) &  Yes  \\ 
			$\mathcal{P}_{\rm fin}(\Sigma^*)$ & No ($f''_1$) &  Yes   \\ 
			$\mathbb{N}_{\max}$ & No ($f_1$) & Yes   \\ 
			$\mathbb{Z}_{\max}$ & No ($f_1$) &  No ($f_2$)   \\ 
			$\mathbb{R}_{\max}$ & No ($f_1$)  &  No ($f_2$)  \\
			
		\end{tabular}
	\end{minipage}
	\caption{Summary of collapse of classes in familiar semirings.}
	\label{fig:table}
\end{figure}

\section{Conclusion and open questions}
\label{section:conclusion}
We have pictured a landscape for learning weighted automata on any semiring and shown that an approach \`a la Angluin can only work in general on a restricted class of functions. This paper opens many questions, both general and specific. First, for functions that are not weakly guessable, one can ask whether there is still a way to find appropriate states (not from rows of the Hankel matrix) by strengthening the queries. Next, for weakly guessable functions, one can ask whether there is a way to design strategies to update $Q$ and $T$ (other than closure and adding suffixes of counter-examples to $T$) that would ensure to find a suitable hypothesis automaton. In particular, is there a way to choose ``cleverly'' the coefficients of the linear combinations under consideration? For example, the algorithm of \cite{HKRS20} can fail for functions that are not guessable (i.e. if for some finite set $Q$ at every Test step the counter example provided by the teacher allows for the possibility that the closure strategy  produces a solution that is not in $\Lambda_Q$). For guessable functions, a naive way to construct an algorithm that terminates with success is to use the update strategy suggested by the proof of Proposition \ref{prop:guessgame}. But is there a better way? And if so, it would be interesting to see if the class of functions learnable by a polynomial time algorithm can be characterised. Regarding the general hierarchy of functions, there is a class for which we neither have an example of non-emptiness or an argument for collapse: the functions that are guessable but not strongly guessable (and its dual). Finally, it would be interesting to study the notion of uniform strong guessability introduced in Remark~\ref{rk:universalsg} and in particular, try to characterise the class of functions lying in the intersection of the six classes we have introduced. It would also be of interest to understand how termination of the algorithm from \cite{HKRS20} relates to strong guessability.  

\section*{Acknowledgment}
This work has been supported by the EPSRC grant EP/T018313/1 and a Turing-Manchester Exchange Fellowship, The Alan Turing Institute. We also thank Nathana\"el Fijalkow for useful discussions on the topic, and the anonymous reviewers for their deep and useful comments.

\bibliographystyle{alphaurl}
\bibliography{ref}

\newpage
\appendix

\section{Dual approach using columns}
\label{section:swap}

In this section, we present (for clarity) the natural dual definitions to those provided in Section~\ref{section:prelim} and Section~\ref{section:guessable} working with (partial) columns instead of (partial) rows and consider right multiplication instead of left multiplication. As a result, the role of final and initial states is swapped in our constructions. $S$ still denotes a semiring and $\Sigma$ an alphabet.

\subsubsection*{\textbf{Right semimodules, closed sets and right generating sets:}} Given a set $Y$, $S^Y$ is a \emph{right} $S$-module via the following action:  for each element $\underline{x} \in S^Y$ and $\lambda$ in $S$, $\underline{x} \otimes \lambda$ denotes the element of $S^Y$ computed from $\underline{x}$ by multiplying every component of $\underline{x}$ by $\lambda$ on the \emph{right}.  Given a subset $X \subseteq S^Y$, a right-linear combination over $X$ is one of the form $\bigoplus_{\underline{x} \in X} (\underline{x} \otimes \lambda_x)$ for some $\lambda_{\underline{x}} \in S$, with only finitely many $\lambda_{\underline{x}}$ different from $0_S$. The right-semimodule generated by $X$ is defined as all the elements of $S^Y$ that can be written as a right-linear combination over $X$. An element of the right-semimodule generated by $X$ will simply be said to be right-generated by $X$. 

As before, we will use the notation $\underline{x}$ for elements of $S^{\Sigma^*}$ and $\underline{x}_Z$ for the restriction of $\underline{x}$ to entries indexed by a subset $Z \subseteq \Sigma^*$, remembering that the action of $S$ is now on the \emph{right}. We also define a \emph{left} action of $\Sigma^*$ on $S^{\Sigma^*}$ as follows: given an element $\underline{x} \in S^{\Sigma^*}$, and $u \in \Sigma^*$, we denote by $u \cdot \underline{x}$ the element of $S^{\Sigma^*}$ defined by $(u\cdot \underline{x})_w = \underline{x}_{wu}$ for all words $w$. It will also be convenient to have a notation for elements of $S^{\Sigma^{*}}$ that arise as columns of a given Hankel matrix $F$. If $v$ is a word, we shall write $[ v ]$ to denote the infinite column indexed by $v$ in
$F$, viewed as an element of $S^{\Sigma^*}$. Note that the left action of $\Sigma^*$ on $[v]$ corresponds to the obvious left action of $\Sigma^*$ on $v$, that is: $(u \cdot [v])_w = [v]_{wu}= F_{wu,v} = f(wuv) = F_{w,uv} = [ uv ]_{w}$. 

Let  $Z \subseteq \Sigma^*$ and $f: \Sigma^* \rightarrow S$. A subset $X \subseteq S^{\Sigma^*}$ is said to be: (i) right-closed if for all $\underline{x}$ in $X$ and $a \in \Sigma$, $a \cdot \underline{x}$ is right-generated by $X$; (ii) right-closed on $Z$ if for all $\underline{x}$ in $X$ and $a \in \Sigma$, $(a\cdot\underline{x})_Z$ is right-generated by the elements of $X$ restricted to $Z$; (iii) column-generating for $f$ if all the columns of the Hankel matrix of $f$ are right-generated by $X$. By an abuse of language, we say that a set of words $T$ is right-closed (resp. right-closed on $Z$, column-generating for $f$) if the set $\underline{T}=\{ [t]: t \in T\}$ has this property.

\subsubsection*{\textbf{Co-Hankel automata and  co-hypothesis automata:}} 
\begin{nota}
\label{nota:gamma}
Given a function $f: \Sigma^* \rightarrow S$, a finite non-empty subset $T$ of $S^{\Sigma^*}$ and a non-empty subset $Q$ of $\Sigma^*$, we define the right-solution set $\Gamma_{Q, T}$ to be set of elements $\underline{\gamma}
\in S^{T \cup (T \times \Sigma \times T)}$ satisfying:
\begin{itemize}
\item $[ \varepsilon ]_Q = \bigoplus_{\underline{t} \in T}  \underline{t}_Q \otimes \gamma_{\underline{t}}$, and
\item for all $\underline{t} \in T$ and all $a \in \Sigma$, $(a \cdot \underline{t})_Q = \bigoplus_{\underline{s} \in T} \underline{s}_Q \otimes \underline{\gamma}_{\underline{s},a,\underline{t}}$.
\end{itemize}\noindent 
If $Q = \Sigma^*$, we will simply write $\Gamma_T$ instead of $\Gamma_{Q, T}$. Moreover, if $R \subseteq \Sigma^*$, then to reduce notation we will also frequently write simply $\Gamma_{Q, R}$ and  $\Gamma_{R}$ instead of $\Gamma_{Q, T}$ and $\Gamma_{T}$   where $T = \{[t] : t \in R\}$.
\end{nota}
It is clear from the definitions above that $\Gamma_{Q,T}$ is non-empty if and only if $[ \varepsilon ]_Q$ is right-generated by $T$ restricted to $Q$ and $T$ is right-closed on $Q$. Also, if $Q'\subseteq Q$, then $\Gamma_{Q', T} \supseteq \Gamma_{Q, T}$. However, if $T'$ is a finite set contained in or contained by $T$, there is no containment between $\Gamma_{Q, T}$ (which lies in $S^{T \cup (T\times \Sigma \times T)}$) and 
$\Gamma_{Q, T'}$ (which lies in $S^{T' \cup (T'\times \Sigma \times T')}$).
\begin{defi}
\label{definition:cohankel}
Let $f: \Sigma^* \rightarrow S$ be a function and suppose that $T$ is a finite subset of $S^{\Sigma^*}$ and $Q$ is a subset of $\Sigma^*$ such that $\Gamma_{Q, T} \neq \emptyset$. Then for each $\gamma \in \Gamma_{Q, T}$, we define a finite weighted automaton $\mathcal{H}^{co}_{Q, T, \underline{\gamma}}$ with:
\begin{itemize}
\item finite set of states $T$;
\item initial-state vector $(\underline{t}_\varepsilon)_{\underline{t} \in T}$;
\item final-state vector $\underline{\gamma}_{T}$; and
\item for all $\underline{s}, \underline{t} \in T$ and all $a \in \Sigma$, a transition from state $\underline{s}$ to state $\underline{t}$ labelled by the letter $a$ with weight $\underline{\gamma}_{\underline{s},a,\underline{t}}$.
\end{itemize}\noindent 
If $Q = \Sigma^*$,  we say $\mathcal{H}^{co}_{Q, T, \underline{\gamma}}$ is a \emph{co-Hankel automaton} and write simply $\mathcal{H}^{co}_{T, \underline{\gamma}}$ to reduce notation. If $T$ is a finite set of columns of the Hankel matrix and $Q$ is finite, we say that $\mathcal{H}^{co}_{Q, T, \underline{\gamma}}$ is a \emph{co-hypothesis automaton}; in this case we may view $T$ as a finite set of words by identifying with an appropriate index set note: there need not be a unique set of words.\end{defi} 

\begin{thm} 
\label{theorem:comain}
Let $f: \Sigma^* \rightarrow S$ be a function. The following are equivalent:
\begin{enumerate}
\item there exists a finite subset $T \subseteq S^{\Sigma^*}$ that is both right-closed and right-generates the column indexed by $\varepsilon$ in the Hankel matrix of $f$ (i.e. $\Gamma_T \neq \emptyset$ for some finite subset $T\subseteq S^{\Sigma^*}$);
\item $f$ is computed by a finite-state weighted automaton $\mathcal{A}$ over $S$.   
\end{enumerate}
Specifically, if $\Gamma_T \neq \emptyset$, then $\mathcal{H}^{\rm co}_{T, \underline{\gamma}}$ computes $f$ for all $\underline{\gamma}\in \Gamma_T$, whilst if $f$ is computed by the finite-state weighted automaton $\mathcal{A}$ over $S$ with state set $T$, then $\Gamma_{\underline{T}} \neq \emptyset$ for $\underline{T} = \{\underline{t}: t \in T\} \subseteq S^{\Sigma^*}$ where for each $w \in \Sigma^*$ we define  $\underline{t}_w$ to be the value on input $w$ computed by the automaton $\mathcal{A}_t$ obtained from $\mathcal{A}$ by making $t$ the unique final state with weight $1_S$ and all other weights are the same).
\end{thm}

\subsubsection*{\textbf{(Weakly/strongly) co-guessable functions and chain conditions:}}
Here we record the definitions used in the statement
\begin{defi}
A function $f: \Sigma^* \rightarrow S$ is said to be weakly co-guessable if there exists a finite set $T$ of columns of the Hankel matrix of $f$ such that $\Gamma_T$ is non-empty.	
\end{defi}

\begin{defi}
	A function $f: \Sigma^* \rightarrow S$ is co-guessable if there exist finite subsets $Q, T \subseteq \Sigma^*$ such that $\Gamma_{Q,T} = \Gamma_{T} \neq \emptyset$. We say that the pair $(Q,T)$ witnesses that $f$ is co-guessable. Furthermore, we say that a co-guessable function $f$ is strongly co-guessable if for all finite sets $T \subseteq \Sigma^*$, there exists a finite set $Q \subseteq \Sigma^*$ such that $\Gamma_{Q,T} = \Gamma_{T}$.
\end{defi}

\begin{defi}
Given a function $f:\Sigma^* \to S$ and $F$ its Hankel matrix, $f$ satisfies the weak co-ascending chain condition if for all chains of right-semimodules $X_0 \subseteq X_1 \subseteq X_2 \subseteq \ldots $ such that
\begin{itemize}
    \item each $X_i$ is generated by a finite set of columns of $F$;
    \item each column of $F$ belongs to some $X_i$
\end{itemize}
there exists $n$ such that for all $m \geq n$, $X_m = X_n$.
\end{defi}

\begin{defi}
Let $T$ be a finite subset of $\Sigma^*$.
A total $T$-chain is an infinite descending chain of right-solution sets $\Gamma_{Q_0, T} \supsetneq \Gamma_{Q_1, T} \supsetneq \Gamma_{Q_2, T} \supsetneq \cdots$  with $Q_i \subseteq Q_{i+1}$ for all $i \geq 0$ and such that each word in  $\Sigma^*$ is contained in some $Q_i$.
We say that a function $f: \Sigma^* \rightarrow S$ is column-bound if there exists a finite set of words $T$ such that 
$\Gamma_T \neq \emptyset$
 and there is no total $T-$chain.
We say that $f: \Sigma^* \rightarrow S$ is strongly column-bound if for \emph{every} finite set of words $T$, there is no total $T$-chain.
\end{defi}

\section{Examples}
\label{appendix:examples}

	\subsubsection*{\textbf{Neither weakly guessable nor weakly co-guessable ($f_1$):}}
	We will give examples of functions that are neither weakly guessable nor weakly co-guessable for $\mathbb{N}_{\max}$, $\mathbb{Z}_{\max}$, $\mathbb{R}_{\max}$, $\mathbb{N}$, $\mathbb{R}_{\geq 0}$ and $\mathcal{P}_{\rm fin}(\Sigma^*)$. We give the proofs for weakly guessable. The proofs for weakly co-guessable are dual.
	
	Let $S$ be one of $\mathbb{N}_{\max}$, $\mathbb{Z}_{\max}$, or $\mathbb{R}_{\max}$, and consider the function $f_1$ computed by the automaton depicted in Figure~\ref{fig:automata1} and $F$ its Hankel matrix over $S$. Thus $f_1(w)$ is equal to the length of the longest block of consecutives $a$'s in $w$. Let $w = ua^n$ and $w'=a^{n'}v$ where $u, v \in\Sigma^*$ are such that $u$ does not end with an $a$ and $v$ does not start with an $a$. The corresponding entry of the Hankel matrix is given by $F_{w,w'} =f(ww') = {\max}\{f_1(u), n+n', f_1(v)\}$. We claim that $f_1$ is neither weakly guessable nor weakly co-guessable. Suppose for contradiction that $f_1$ is weakly guessable. By Theorem~\ref{prop:weakly} there exists a finite set of words $W=\{w_1, \ldots ,w_k\}$ such that the rows of the Hankel matrix indexed by $W$ are row-generating. Let us write $w_i = u_ia^{n_i}$ for $i=1, \ldots ,k$, where $u_i$ does not end with an $a$ and $n_i \geq 0$, and choose a positive integer $N$ such that $N\geq {\max}_{i \in \{1,\ldots,k\}}\{f_1(u_i), n_i\}$. Consider the row of the Hankel matrix indexed by $w=a^{N+1}b$. Since the rows indexed by $W$ are assumed to be row-generating, we have:
	$\langle w \rangle = \bigoplus_{i=1}^k \lambda_i \otimes \langle w_i \rangle$ for some $\lambda_i \in S$. Now set $u = ba^{N}$ and $v=ba^{N+1}$. We have:
	\begin{eqnarray*}
		f_1(wu) = F_{w,u} = \langle w \rangle_u &=& {\max}_{i}\{\lambda_i + \langle w_i \rangle_ u\}\\
		&=& {\max}_{i}\{\lambda_i + F_{w_i, u}\}\\
		&=& {\max}_{i}\{\lambda_i + \max(f_1(u_i), n_i, N)\}\\
		&=& {\max}_{i}\{\lambda_i + N\} = {\max}_{i}\{\lambda_i\} + N  
	\end{eqnarray*} 
	and, similarly, $f_1(wv) = {\max}_{i}\{\lambda_i\} + N +1$.  But this immediately gives a contradiction, since $f_1(wu) = f_1(a^{N+1}b^2a^N) = N+1 = f_1(a^{N+1}b^2a^{N+1}) = f_1(wv)$.
	
	Let $S$ be one of $\mathbb{N}$, $\mathbb{R}_{\geq 0}$ and consider the function $f'_1$ computed by the automaton depicted in Figure~\ref{fig:automata4} and $F$ its Hankel matrix. Thus $f'_1(w)$ is equal $2^n - 1$ if $w=a^n$ for some $n>0$ and $0$ otherwise. Suppose that $f'_1$ is weakly guessable. Then there exist a finite set of words $Q = \{w_1, w_2, \ldots, w_k\}$ such that the rows indexed by $Q$ left-generate the rows of $F$. Without loss of generality, we may assume that $w_1 = \varepsilon$ and that $w_i$ is a power of $a$ for $i=1, \ldots m \leq k$. Let $w = a^{n}$ for some positive integer $n$. Then by assumption, considering columns $\varepsilon$ and $a$ of $F$ and restricting attention to row $w$, there exist $\lambda_1, \ldots, \lambda_k$ in $S$ such that:
	\[2^{n} - 1 = \sum_{i=1}^k \lambda_i f'_1(w_i) \mbox{ and  }2^{n+1} - 1 = \sum_{i=1}^k \lambda_i f'_1(w_ia),\] which gives $\sum_{i=1}^k \lambda_i f'_1(w_ia) = 1+  \sum_{i=1}^k 2\lambda_i f'_1(w_i).$ Moreover, for each $w_i$ that is a power of $a$, we have $f'_1(w_ia) = 2f'_1(w_i) + 1$, whilst if $w_i$ is not a power of $a$, we have $f'_1(w_ia) = 0 = f'(w_i)$. Thus when restricting attention to columns $\{\varepsilon, a\}$ it suffices to consider only the $w_i$ with $i \leq m$, giving:
	\[\sum_{i=1}^m \lambda_i (2f'_1(w_i) + 1) = 1 + \sum_{i=1}^m 2\lambda_i f'_1(w_i)\]
	giving $\sum_{i=1}^m \lambda_i=1$.
	
	Suppose now that $n$ is (strictly) larger than the lengths of all the $w_i$. We have proved that there are $\lambda_i \geq 0$ such that $2^{n} - 1 = \sum_{i=1}^m \lambda_i f'_1(w_i)$ and $\sum_{i=1}^m \lambda_i=1$. But $2^n - 1$ is strictly greater than any of the $f'_1(w_i)$, so cannot be obtained as a convex combination of them. 
	
	For $S=\mathcal{P}_{\rm fin}(\Sigma^*)$, consider the function $f''_1$ computed by the automaton depicted in Figure~\ref{fig:automata4bis} and $F$ its Hankel matrix. Thus $f''_1(w) = \{a^{|w|_a},b^{|w|_b}\}$. If a row $\langle w \rangle$ of $F$  is left-generated by a finite set of rows, this in particular means (by restricting attention to column $\varepsilon$) that $\{a^{|w|_a}, b^{|w|_b}\}$ can be written as $\bigcup_{i=1}^k X_i\{a^{n_i}, b^{m_i}\}$ for some  finite sets $X_i$ and non-negative integers $n_i,m_i$. This is only possibly if exactly one of the sets $X_i$ is non-empty, and for this value $i$ we have that $X_i = \{\varepsilon\}$ and $n_i = |w|_a$ and  $m_i = |w|_b$. It follows from this that for all pairs of non-negative integers $(n,m)$ a row-generating set of words must contain at least one word with $|w|_a=n$ and $|w|_b=m$.
	\vspace{-1\baselineskip}
	\subsubsection*{\textbf{Weakly guessable but neither guessable nor weakly co-guessable ($f_2$)}.}
\label{app:examplef2}
	Let $S=\mathbb{Z}_{\max}$ or $S=\mathbb{R}_{\max}$ and consider the function $f_2$  computed by the automaton depicted in Figure~\ref{fig:automata3}, but viewed as an automaton over $S$ and let $F$ denote its Hankel matrix. Thus $f_2(w) = |w|_a$ if $w$ starts with an $a$, $f_2(w) = |w|_b$ if $w$ starts with an $b$, and $f_2(\varepsilon) = -\infty$.  Recalling that multiplication over this semiring 
    is given by usual addition, if $w$ starts with $a$ we have $\langle w \rangle = (|w|_a - 1) \otimes \langle a \rangle$. Likewise, if $w$ starts with $b$ we have $\langle w \rangle = (|w|_b - 1) \otimes \langle b \rangle$. It follows from these simple observations that every row of the Hankel matrix is a scalar multiple of one of the rows $\langle \varepsilon \rangle, \langle a \rangle$ or $\langle b \rangle$. Moreover, it is not too difficult to see that none of these three can be expressed as a linear combination of the other two: for example the fact that $f_2(\varepsilon)= -\infty$ and $f_2(u) \neq -\infty$ for all non-empty words $u$ quickly demonstrates that $\langle \varepsilon \rangle$ cannot be generated by any combination of the other rows.
	
	First, $f_2$ is weakly guessable. Indeed, it is easy to see that one can transform the automaton given in the figure into an equivalent literal automaton, by pushing the weights from the first transitions to final weights. 
	
	Second, $f_2$ is not weakly co-guessable. For all words $w$, we have $f_2(aw) = |aw|_a$, and $f_2(bw) = |bw|_b$. Suppose for contradiction that $f_2$ is weakly co-guessable. Then there exists a finite set of words $W=\{w_1, \ldots ,w_k\}$ such that the columns of the Hankel matrix labelled by $W$ are column-generating. Let $N> \max_i|w_i|_a + 1$. Since for all $i$, we have $F_{a, w_i} = |w_i|_a + 1$ and $F_{b, w_i} = |w_i|_b+1$, it is easy to see that the column indexed by $w = a^N$ (for which $F_{a, w} = N+1$ and $F_{b, w} = 1$) cannot be expressed as a right-linear combination of the columns indexed by $W$.
	
		Finally, $f_2$ is not guessable. Suppose for contradiction that $f_2$ is guessable. Then there exist finite sets $Q$, $T$ such that $\Lambda_{Q, T} = \Lambda_Q \neq \emptyset$, and hence for all $\underline{\lambda}$ in $\Lambda_{Q,T}$, we have that $\mathcal{H}_{Q,T,\lambda}$ computes $f_2$. By the observations at the start of this subsection, we must have that $Q$ contains $\varepsilon$ and at least one word, say $q$, beginning with $a$ and at least one word, say $q'$,  beginning with $b$. Without loss of generality, we may assume that for a fixed constant $N$ (to be chosen below) we have that the length of the longest word in $T$ is greater than $N$ (since for all $T' \supseteq T$ we also have $\Lambda_{Q, T} = \Lambda_{Q, T'}= \Lambda_Q \neq \emptyset$). Now let $n$ be equal to the maximal sum of the length a word in $Q$ and the length of a word in $T$ (which by the above observation can be made arbitrarily large, by modifying $T$ if necessary). For this choice of $n$ and for all $t \in T$ we have
		\begin{align*} 
		&\langle qa \rangle_t \!\!\!&=& f_2(qat) \!\!\!&=& |qat|_a \!\!\!&=& |qt|_a +1 \geq 0 \geq |q't|_b - n \!\!\!&\Rightarrow& \langle qa \rangle_T  \!\!\!&=& (1 \otimes \langle q \rangle_T) \oplus (-n \otimes \langle q' \rangle_T)\\
		&\langle qb \rangle_t \!\!\!&=& f_2(qbt) \!\!\!&=& |qbt|_a \!\!\!&=& |qt|_a + 0  \geq 0 \geq |q't|_b - n \!\!\!&\Rightarrow& \langle qb \rangle_T \!\!\!&=& (0 \otimes \langle q \rangle_T) \oplus, (-n \otimes \langle q' \rangle_T)\\
		&\langle q'a \rangle_t \!\!\!&=& f_2(q'at) \!\!\!&=& |q'at|_b \!\!\!&=& |qt|_b +0 \geq 0 \geq |qt|_a - n \!\!\!&\Rightarrow& \langle q'a \rangle_T \!\!\!&=& (-n \otimes \langle q \rangle_T) \oplus 0 \langle q' \rangle_T\\
		&\langle q'b \rangle_t \!\!\!&=& f_2(q'bt) \!\!\!&=& |q'bt|_b \!\!\!&=& |q't|_b +1 \geq 0 \geq |qt|_a-n \!\!\!&\Rightarrow& \langle q'b \rangle_T \!\!\!&=& (-n \otimes \langle q \rangle_T) \oplus (1 \otimes\langle q' \rangle_T).
		\end{align*}
         Now let $\mathcal{H_{Q,T, \underline{\lambda}}}$ be the hypothesis automaton constructed from an element $\underline{\lambda} \in \Lambda_{Q, T} = \Lambda_Q$ using the linear combinations above (containing the \emph{negative} coefficients $-n$) and let $\mathcal{H}_q$ denote the automaton obtained from $\mathcal{H}_{Q, T, \underline{\lambda}}$ by making $\langle q \rangle$ the unique initial state. Since $\underline{\lambda} \in \Lambda_Q$, it follows from the first part of the proof of Theorem \ref{theorem:prelim} that for all words $w$, and all $q \in Q$, $\langle q \rangle_w$ is equal to the value computed by $\mathcal{H}_q$ on input $w$. Thus in particular we have that $f_2(qa^nbb^{2n})$ is equal to the value computed by $\mathcal{H}_q$ on input $a^nbb^{2n}$. Consider the run composed of: the transition from $q$ to $q$ labelled by $a$ for each instance of $a$ in the factor $a^n$, the transition from $q$ to $q'$ labelled by $b$ for the first factor of $b$, and the transition from $q'$ to $q'$ labelled by $b$  for each instance of $b$ in the remaining factor $b^{2n}$. Since the final weight at state $q'$ is (by definition) $\langle q' \rangle_\varepsilon = f_2(q') = |q'|_b$, it then follows from the equations above (recalling that $\otimes$ denotes usual addition in the semirings under consideration) that this run has weight $n+  (-n) +n + |q'|_b = 2n + |q'|_b $ giving $f_2(qa^nbb^{2n}) > 2n$ (recalling that the addition of this semiring is maximisation). On the other hand, since $q$ starts with the letter $a$, we have $f_2(qa^nbb^{2n}) = |q|_a + n$ (from the definition of this function). But now taking $N=|q|_a$ we have $f_2(qa^nbb^{2n}) > 2n> n + N = |q|_a + n = f_2(qa^nbb^{2n})$, giving a contradiction.
	\subsubsection*{\textbf{Weakly co-guessable but neither co-guessable nor weakly guessable ($\bar{f_2}$):}}
	Since $\mathbb{Z}_{\max}$ and $\mathbb{R}_{\max}$ are commutative, and $f_2$ is weakly guessable but neither guessable nor weakly co-guessable, it follows from Lemma \ref{lem:com2} that  $\bar{f_2}$ is weakly co-guessable but neither co-guessable nor weakly guessable (since a function over a commutative semiring is (weakly) guessable if and only if its mirror is (weakly) co-guessable).
	
	\subsubsection*{\textbf{Strongly guessable but not weakly co-guessable ($f_3$ and $f_5$):}}
    \label{app:examplef3}
	We will give examples of functions that are strongly guessable but not weakly co-guessable for $\mathbb{N}_{\max}$, $\mathbb{N}$, $\mathcal{P}_{\rm fin}(\Sigma^*)$ and $\mathbb{R}_{\geq 0}$. In Lemma~\ref{lemma:PM}, it is shown that all the functions from the first three semirings are strongly row-bound and hence, any weakly guessable function is strongly guessable. It is also the case for $\mathbb{R}_{\geq 0}$ by Lemma \ref{prop:subsemiring}. To prove that the functions under consideration are strongly guessable, it is then enough to prove that they are weakly guessable, or equivalently that they are computed by a literal automaton. 
	
	For  $S=\mathbb{N}_{\max}$, consider the function $f_3$ computed by the automaton depicted in Figure~\ref{fig:automata3} and $F$ its Hankel matrix. Thus $f_3(w) = |w|_a$ if $w$ starts with an $a$, $f_3(w) = |w|_b$ if $w$ starts with an $b$, and $f_3(\varepsilon) = -\infty$, and (by the same argument as for $f_2$) we find that $f_3$ is computed by a literal automaton and weakly guessable; in this case however Lemma~\ref{lemma:PM}  applies to give that $f_3$ is also strongly guessable (and note that the proof that $f_2$ is \emph{not} guessable made use of \emph{negative} coefficents that we do not have access to here!)) Let us prove now that it is not weakly co-guessable. For all words $w$, we have $f_3(aw) = |aw|_a$, and $f_3(bw) = |bw|_b$. Suppose for contradiction that $f_3$ is weakly co-guessable. Then there exists a finite set of words $W=\{w_1, \ldots ,w_k\}$ such that the columns of the Hankel matrix labelled by $W$ are column-generating and let $N> \max_i|w_i|_a + 1$. Since for all $i$, we have $F_{a, w_i} = |w_i|_a + 1$ and $F_{b, w_i} = |w_i|_b+1$, it is easy to see that the column indexed by $w = a^N$ (for which $F_{a, w} = N+1$ and $F_{b, w} = 1$) cannot be expressed as a right-linear combination of the columns indexed by $W$.
	
	For $S=\mathbb{N}$, consider the function $f'_3$ computed by the automaton depicted in Figure~\ref{fig:automata5} and $F$ its Hankel matrix. Thus $f'_3(w) = 2^{|w|_a}$ if $w$ starts with an $a$, $f_3(w) = 2^{|w|_b}$ if $w$ starts with an $b$, and $f_3(\varepsilon) = 0$. It is easy to see that $f'_3$ is computed by a literal automaton, by pushing the weights on the first transitions in the automaton depicted in the figure, as final weights. So $f'_3$ is weakly guessable and hence strongly guessable. Let us prove now that it is not weakly co-guessable. For all words $w$, we have $f'_3(aw) = 2^{|aw|_a}$, and $f'_3(bw) = 2^{|bw|_b}$. Suppose for contradiction that $f'_3$ is weakly co-guessable. Then there exists a finite set of words $W=\{w_1, \ldots ,w_k\}$ such that the columns of the Hankel matrix labelled by $W$ are column-generating and let $N> \max_i|w_i|_a + 1$. Since for all $i$, we have $F_{a, w_i} = 2^{|w_i|_a + 1}$ and $F_{b, w_i} = 2^{|w_i|_b+1}$, it is easy to see that the column indexed by $w = a^N$ (for which $F_{a, w} = 2^{N+1}$ and $F_{b, w} = 2$) cannot be expressed as a right-linear combination of the columns indexed by $W$.
	
	For $S=\mathcal{P}_{\rm fin}(\Sigma^*)$, consider the function $f''_3$ computed by the automaton depicted in Figure~\ref{fig:automata5bis} and $F$ its Hankel matrix. Thus $f''_3(w) = \{x^{|w|_x}\}$ if $w$ starts with an $x \in \{a,b\}$, and $f''_3(\varepsilon) = \emptyset$. It is easy to see that $f''_3$ is computed by a literal automaton, by pushing the weights on the first transitions in the automaton depicted in the figure, as final weights. So $f''_3$ is weakly guessable and hence strongly guessable. Let us prove now that it is not weakly co-guessable. For all words $w$, we have $f''_3(aw) = a^{|aw|_a}$, and $f''_3(bw) = b^{|bw|_b}$. Suppose for contradiction that $f''_3$ is weakly co-guessable. Then there exists a finite set of words $W=\{w_1, \ldots ,w_k\}$ such that the columns of the Hankel matrix labelled by $W$ are column-generating and let $N> \max_i|w_i|_a + 1$. Since for all $i$, we have $F_{a, w_i} = a^{|w_i|_a + 1}$ and $F_{b, w_i} = b^{|w_i|_b+1}$, it is easy to see that the column indexed by $w = a^N$ (for which $F_{a, w} = a^{N+1}$ and $F_{b, w} = b$) cannot be expressed as a right-linear combination of the columns indexed by $W$.
	
	Finally, we give an example of a function over $S=\mathbb{R}_{\geq 0}$ that is strongly guessable but not weakly co-guessable. Consider the automaton obtained from that in Figure~\ref{fig:automata5} by replacing all weights labelled by $b$ by $1$. That is, the function computed, which we denote by $f_5$, is defined by $f_5(w) = 2^{|w|_a}$ if $w$ begins with $a$,  $f_5(w) = 1$ if $w$ begins with $b$ and $f_5(\varepsilon) =0$. In a similar manner to the previous cases, it is easy to see that $f_5$ is computed by a literal automaton and hence is weakly guessable.  Taking $Q=T=\{\varepsilon, a, b\}$ it is straightforward to verify that $\Lambda_{Q, T} = \Lambda_Q \neq \emptyset$ (indeed, for each of $\langle \varepsilon \rangle, \langle q a\rangle, \langle qb \rangle$ where $q \in Q$ we find that there is exactly one way to write the given vector as a linear combination of the $\langle q' \rangle$ with $q' \in Q$). Hence $f_5$ is strongly guessable. Suppose then for contradiction that $f_5$ is weakly co-guessable. Then there exists a finite set of words $W=\{w_1, \ldots ,w_k\}$ such that the columns of the Hankel matrix labelled by $W$ are column-generating and let $N> \sum_i2^{|w_i|_a + 1}$. Since for all $i$, we have $F_{a, w_i} = 2^{|w_i|_a + 1}$ and $F_{b, w_i} = 1$, it is easy to see that the column indexed by $w = a^N$ (for which $F_{a, w} = 2^{N+1}$ and $F_{b, w} = 1$) cannot be expressed as a right-linear combination of the columns indexed by $W$.

		\subsubsection*{\textbf{Strongly co-guessable but not weakly guessable ($\bar{f_3}$):}}
		Since $\mathbb{N}_{\max}$ and $\mathbb{N}$ are commutative, and $f_3$ and $f'_3$ are strongly guessable but not weakly co-guessable, then (by Lemma \ref{lem:com2}) $\bar{f_3}$ and $\bar{f'_3}$ are strongly co-guessable but not weakly guessable .
\end{document}